\documentclass{amsart}%
\usepackage{amsfonts}
\usepackage{amsmath}
\usepackage{amssymb}
\usepackage{epsfig}
\usepackage{graphicx}%
\newtheorem{theorem}{Theorem}
\theoremstyle{plain}
\newtheorem{acknowledgement}{Acknowledgement}

\newtheorem{corollary}{Corollary}

\newtheorem{definition}{Definition}
\newtheorem{example}{Example}

\newtheorem{lemma}{Lemma}
\newtheorem{notation}{Notation}

\newtheorem{proposition}{Proposition}
\newtheorem{remark}{Remark}

\numberwithin{equation}{section}
\begin{document}
\title[Graphs, Zeta Functions, and Log-Coulomb Gases]{Graphs, local zeta functions, Log-Coulomb Gases, and phase transitions at
finite temperature}
\author[Z\'{u}\~{n}iga-Galindo]{W. A. Z\'{u}\~{n}iga-Galindo}
\address{University of Texas Rio Grande Valley\\
School of Mathematical \& Statistical Sciences\\
One West University Blvd\\
Brownsville, TX 78520, United States and Centro de Investigaci\'{o}n y de
Estudios Avanzados del Instituto Polit\'{e}cnico Nacional\\
Departamento de Matem\'{a}ticas, Unidad Quer\'{e}taro\\
Libramiento Norponiente \#2000, Fracc. Real de Juriquilla. Santiago de
Quer\'{e}taro, Qro. 76230\\
M\'{e}xico.}
\email{wilson.zunigagalindo@utrgv.edu, wazuniga@math.cinvestav.edu.mx}
\author[Zambrano-Luna]{B. A. Zambrano-Luna}
\address{Centro de Investigaci\'{o}n y de Estudios Avanzados del Instituto
Polit\'{e}cnico Nacional\\
Departamento de Matem\'{a}ticas, Av. Instituto Polit\'{e}cnico Nacional \#
2508, Col. San Pedro Zacatenco, CDMX. CP 07360 \\
M\'{e}xico.}
\email{bazambrano@math.cinvestav.mx}
\author[Le\'{o}n-Cardenal]{Edwin Le\'{o}n-Cardenal}
\address{CONACYT -- Centro de Investigaci\'{o}n en Matem\'{a}ticas, Unidad Zacatecas.
Quantum Ciudad del Conocimiento. Avenida Lasec, Andador Galileo Galilei,
Manzana 3 Lote 7. C.P. 98160. Zacatecas, ZAC. \\
M\'{e}xico.}
\email{edwin.leon@cimat.mx}

\begin{abstract}
We study a log-gas on a network (a finite, simple graph) confined in a bounded
subset of a local field (i.e. $\mathbb{R}$, $\mathbb{C}$, $\mathbb{Q}_{p}$ the
field of $p$-adic numbers). In this gas, a \ log-Coulomb interaction\ between
two charged particles occurs only when the sites of the particles are
connected by an edge of the network. The partition functions of such gases
turn out to be a particular class of multivariate local zeta functions
attached to the network and a positive test function which is determined by
the confining potential. The methods and results of the theory of local zeta
functions allow us to establish that the partition functions admit meromorphic
continuations in the parameter $\beta$ (the inverse of the absolute
temperature). We give conditions on the charge distributions and the confining
potential\ such that the meromorphic continuations of the partition functions
have a pole at a positive value $\beta_{UV}$, which implies the existence of
phase transitions at finite temperature. In the case of $p$-adic fields the
meromorphic continuations of the partition functions are rational functions
\ in the variable $p^{-\beta}$. We give an algorithm for computing such
rational functions. For this reason, we can consider the $p$-adic log-Coulomb
gases as exact solvable models. We expect that all these models for different
local fields share common properties, and that they can be described by a
uniform theory.

\end{abstract}
\maketitle

\section{Introduction}

In this article we study log-Coulomb gases on finite simple graphs confined in
bounded regions. The partition function of these gases are local zeta
functions (in the sense of Gel'fand, Atiyah, Igusa, Denef, Loeser, among
others). By using the theory of local zeta functions, we establish the
existence of phase transitions at finite temperature. \ The coordinates of the
sites having the charged particles can be taken from any local field
$\mathbb{K}$, for instance $\mathbb{R}$, $\mathbb{C}$, $\mathbb{Q}_{p}$.

An ultrametric space $(M,d)$ is a metric space $M$ with a distance satisfying
the strong triangle inequality $d(A,B)\leq\max\left\{  d\left(  A,C\right)
,d\left(  B,C\right)  \right\}  $ for any three points $A$, $B$, $C$ in $M$.
The field of $p$-adic numbers $\mathbb{Q}_{p}$ constitutes a central example
of an ultrametric space. The ultrametricity, which is the emergence of
ultrametric spaces in physical models, was discovered in the middle 1980s by
Parisi et al. in the context of the spin glass theory, see e.g. \cite{M-P-V},
\cite{R-T-V}. Ultrametric spaces constitute the right framework to formulate
models where hierarchy plays a central role. Ultrametric models have been
applied in many areas, including, quantum physics, $p$-adic string theory,
$p$-adic Feynman integrals, brain and mental states models, relaxation of
complex systems, evolutionary dynamics, cryptography and geophysics, among
other areas, see e.g. \cite{Anashin}, \cite{Av-4}-\cite{Av-5},
\cite{Bocardo-Gaspar et al}, \cite{D-K-K-V}, \cite{Khrenikov1}%
-\cite{Khrennikov2}, \cite{KKZuniga}-\cite{Mukhamedov-3}, \cite{Parisi}%
-\cite{Parisi-Sourlas}, \cite{V-V-Z}-\cite{Zuniga-Torba}, and the references therein.

The Ising models over ultrametric spaces have been studied intensively, see
e.g. \cite{DysonFreeman}, \cite{Gubser et al}, \cite{Khrennikov et al},
\cite{Khrennikov2}, \cite{Lerner-Missarov}, \cite{Mis}, \cite{Mukhamedov-1},
\cite{Mukhamedov-2}, \cite{Mukhamedov-3}, \cite{Sinai} and the references
therein. An important motivation comes from the hierarchical Ising model
introduced in \cite{DysonFreeman}. The hierarchical Hamiltonian introduced by
Dyson in \cite{DysonFreeman} can be naturally studied in $p$-adic spaces, see
e.g. \cite{Lerner-Missarov}, \cite{Gubser et al}. In \cite{Parisi-Sourlas},
see also \cite{Khrennikov-Kozyrev}, Parisi and Sourlas presented a $p$-adic
formulation of replica symmetry breaking. In this approach ultrametricity is a
natural consequence of the topology of the $p$-adic numbers. This work raises
the problem of knowing if it is possible to have a rigorous $p$-adic
formulation of the replication method. This requires, among other things, a
rigorous mathematical understanding of objects such as partition functions in
a $p$-adic framework. This is precisely the objective of the present work.
This article continues the investigation on $p$-adic Coulomb gases started in
\cite{Zuniga-Torba}.

The log-Coulomb gases in the Archimedean context has been studied extensively
in the case of complete graphs see e.g. \cite{Forrester}, \cite{Simon} and the
references therein. The case of arbitrary graphs seems completely new to the
best of our knowledge. Our results suggest that an adelic formulation of the
log-Coulomb gases seems feasible. On the other hand, in this article we study
log-Coulomb gases on finite simple graphs confined in the $p$-adic balls of
arbitrary dimension. A natural problem is to study these models in the
Archimedean framework and to compare them with the non-Archimedean
counterparts. To the best of our knowledge this type of systems has not been
study yet.

By a generalized Mehta integral, we mean an integral of the form%
\[
Z_{\varphi}(\boldsymbol{s})=\int\limits_{-\infty}^{\infty}\cdots
\int\limits_{-\infty}^{\infty}\varphi\left(  x_{1},\ldots,x_{N}\right)
\prod\limits_{1\leq i<j\leq N}\left\vert x_{i}-x_{j}\right\vert ^{s_{ij}}%
\prod\limits_{i=1}^{N}dx_{i},
\]
where $\varphi$ is a Schwartz function, and $\boldsymbol{s=}\left(
s_{ij}\right)  _{1\leq i<j\leq N}\in\mathbb{C}^{\frac{N(N-1)}{2}}$ with
$\operatorname{Re}\left(  s_{ij}\right)  >0$ for any $1\leq i<j\leq N$. The
original Mehta integral $F_{N}(\gamma)$ is exactly $F_{N}(\gamma)=\frac
{1}{\left(  2\pi\right)  ^{\frac{N}{2}}}Z_{\varphi}(\boldsymbol{s}%
)\mid_{s_{ij}=2\gamma}$, with $\varphi\left(  x_{1},\ldots,x_{N}\right)  =$
$e^{-\frac{1}{2}\sum_{i=1}^{N}x_{i}^{2}}$, and it is the partition function of
a $1$D log-Coulomb gas, see e.g. \cite{Forrester}, \cite{Forrester et al}. The
integral $Z_{\varphi}(\boldsymbol{s})$ is a particular case of a multivariate
local zeta function. These functions admit meromorphic continuations to the
whole $\mathbb{C}^{\frac{N(N-1)}{2}}$, see e.g. \cite{Loeser} . Nowadays,
there exists a uniform theory of local zeta functions over local fields of
characteristic zero, e.g. $\left(  \mathbb{R},\left\vert \cdot\right\vert
\right)  $, $\left(  \mathbb{C},\left\vert \cdot\right\vert \right)  $, and
the field of $p$-adic numbers $\left(  \mathbb{Q}_{p},\left\vert
\cdot\right\vert _{p}\right)  $, see \cite{Igusa-old}, \cite{Igusa}, see also
\cite{Denef}, \cite{DL1}, \cite{G-S}, \cite{Loeser}, \cite{Veys-Zuniga} and
the references therein. By using this theory, we can construct incarnations of
the integral $Z_{\varphi}(\boldsymbol{s})$ over $\mathbb{C}$ and
$\mathbb{Q}_{p}$, which admit meromorphic continuations to the whole
$\mathbb{C}^{\frac{N(N-1)}{2}}$. In addition, the possible poles of all these
functions can be described in a geometric way.

Given a local field $(\mathbb{K},\left\vert \cdot\right\vert _{\mathbb{K}})$
and a finite, simple graph $G$, we attach \ to them a $1$D log-Coulomb gas and
a local zeta function. By a gas configuration\ we\ mean a triple $\left(
\boldsymbol{x},\boldsymbol{e},G\right)  $, with $\boldsymbol{x}=\left(
x_{v}\right)  _{v\in V(G)}$, $\boldsymbol{e}=\left(  e_{v}\right)  _{v\in
V(G)}$, where $e_{v}\in\mathbb{R}$ is a charge located at the site $x_{v}%
\in\mathbb{K}$, and the interaction between the charges is determined by the
graph $G$. Given a vertex $u$ of $G$ ($u\in V(G)$), the charged particle at
the site $x_{u}$ can interact only with those particles located at sites
$x_{v}$ for which there exists an edge between $u$ and $v$ (we denote this
fact as $u\sim v$). The Hamiltonian is given by
\begin{equation}
H_{\mathbb{K}}(\boldsymbol{x};\mathbf{e},\beta,\Phi,G)=-%
%TCIMACRO{\tsum \limits_{\substack{u,v\in V\left(  G\right)  \\u\sim v}}}%
%BeginExpansion
{\textstyle\sum\limits_{\substack{u,v\in V\left(  G\right)  \\u\sim v}}}
%EndExpansion
\ln\left\vert x_{u}-x_{v}\right\vert _{\mathbb{K}}^{e_{u}e_{v}}+\frac{1}%
{\beta}P(\boldsymbol{x}), \label{Hamiltonian}%
\end{equation}
where $\beta=\frac{1}{k_{B}T}$ (with $k_{B}$ the Boltzmann constant, $T$ the
absolute temperature), $P:\mathbb{K}^{\left\vert V(G)\right\vert
}\mathbb{\rightarrow R}$ is a confining potential such that $\Phi\left(
\boldsymbol{x}\right)  =e^{-P(\boldsymbol{x})}$ is a test function, which
means that $P=+\infty$ outside of a compact subset.

The partition function attached to the Hamiltonian (\ref{Hamiltonian}) is
given by%
\begin{equation}
\mathcal{Z}_{G,\mathbb{K},\Phi,\mathbf{e}}\left(  \beta\right)  =%
%TCIMACRO{\dint \limits_{\mathbb{K}^{\left\vert V(G)\right\vert }}}%
%BeginExpansion
{\displaystyle\int\limits_{\mathbb{K}^{\left\vert V(G)\right\vert }}}
%EndExpansion
\text{ }\Phi\left(  \boldsymbol{x}\right)
%TCIMACRO{\tprod \limits_{_{\substack{u,v\in V\left(  G\right)  \\u\sim v}}}}%
%BeginExpansion
{\textstyle\prod\limits_{_{\substack{u,v\in V\left(  G\right)  \\u\sim v}}}}
%EndExpansion
\ \left\vert x_{u}-x_{v}\right\vert _{\mathbb{K}}^{e_{u}e_{v}\beta}%
%TCIMACRO{\tprod \limits_{v\in V\left(  G\right)  }}%
%BeginExpansion
{\textstyle\prod\limits_{v\in V\left(  G\right)  }}
%EndExpansion
dx_{v}. \label{Partition_function}%
\end{equation}
In order to study this integral, using geometric techniques, it is convenient
to extend $e_{u}e_{v}\beta$ to a complex variable $s\left(  u,v\right)  $, in
this way the partition function (\ref{Partition_function}) becomes a local
zeta function. Then the partition function is recovered from the local zeta
function taking $s\left(  u,v\right)  =e_{u}e_{v}\beta$.

The local zeta function attached to $G$, $\Phi$ is defined as
\[
Z_{\Phi}(\boldsymbol{s};G,\mathbb{K})=\int\limits_{\mathbb{K}^{\left\vert
V(G)\right\vert }}\text{ }\Phi\left(  \boldsymbol{x}\right)  \prod
\limits_{\substack{u,v\in V(G)\\u\sim v}}\left\vert x_{u}-x_{v}\right\vert
_{\mathbb{K}}^{s\left(  u,v\right)  }\prod\limits_{v\in V(G)}dx_{v}\text{,}%
\]
where $\boldsymbol{s}=\left(  s\left(  u,v\right)  \right)  $ for $u,v\in
V(G)$ for $u\sim v$, $s\left(  u,v\right)  $ is a complex variable attached to
the edge connecting the vertices $u$ and $v$, and $%
%TCIMACRO{\tprod \nolimits_{v\in V(G)}}%
%BeginExpansion
{\textstyle\prod\nolimits_{v\in V(G)}}
%EndExpansion
dx_{v}$ is a Haar measure of the locally compact group $(\mathbb{K}%
^{\left\vert V(G)\right\vert },+)$. The integral converges for
$\operatorname{Re}(s\left(  u,v\right)  )>0$ for any $\left(  u,v\right)
$.\ The partition function $\mathcal{Z}_{G,\mathbb{K},\Phi,\mathbf{e}}\left(
\beta\right)  $ of $H_{\mathbb{K}}(\boldsymbol{x};\mathbf{e},\beta,\Phi,G)$ is
related to the local zeta function of the graph by%
\[
\mathcal{Z}_{G,\mathbb{K},\Phi,\mathbf{e}}\left(  \beta\right)  =\left.
Z_{\Phi}(\boldsymbol{s};G,\mathbb{K})\right\vert _{s\left(  u,v\right)
=e_{u}e_{v}\beta}.
\]
The zeta function $Z_{\Phi}(\boldsymbol{s};G)$ admits a meromorphic
continuation to the whole complex space $\mathbb{C}^{\left\vert
E(G)\right\vert }$, see \cite[Th\'{e}or\`{e}me 1.1.4]{Loeser}.

For a charge configuration $\boldsymbol{e}=\left(  e_{v}\right)  _{v\in V(G)}$
satisfying that $e_{u}e_{v}>0$ for any $u\sim v$, the partition function
$\mathcal{Z}_{G,\mathbb{K},\Phi,\mathbf{e}}\left(  \beta\right)  $ is analytic
for $\beta>0$. \ If the sign of $e_{u}e_{v}$, for $u\sim v$, changes along the
graph, then the partition function becomes an integral of a `rational
function' on a compact subset, and in the general case, the analyticity for
$\beta>0$ does not hold anymore. The existence of a meromorphic continuation
for $\mathcal{Z}_{G,\mathbb{K},\Phi,\mathbf{e}}\left(  \beta\right)  $ having
positive poles, say at \ $\beta=\beta_{UV}>0$, implies that the function
$\ln\mathcal{Z}_{G,\mathbb{K},\Phi,\mathbf{e}}\left(  \beta\right)  $ has a
pole at $\beta=\beta_{UV}$, and thus any canonical free energy defined using
$\ln\mathcal{Z}_{G,\mathbb{K},\Phi,\mathbf{e}}\left(  \beta\right)  $ has a
pole at $\beta=\beta_{UV}$. Notice that the existence of such a pole does not
require to pass to the thermodynamic limit. Since the canonical energy is not
analytic around $\beta=\beta_{UV}$, this point is a phase-transition point. We
will say that $\mathcal{Z}_{G,\mathbb{K},\Phi,\mathbf{e}}\left(  \beta\right)
$ has a phase transition at temperature $\frac{1}{k_{B}\beta_{UV}}$. The
determination of the actual poles for $Z_{\Phi}(\boldsymbol{s};G,\mathbb{K})$
is a difficult open problem. If $\mathbb{K}$ is a $p$-adic field then
$\mathcal{Z}_{G,\mathbb{K},\Phi,\mathbf{e}}\left(  \beta\right)  $ admits a
meromorphic continuation as a rational function in the variables
$p^{-e_{u}e_{v}\beta}$, $u\sim v$. For this reason we can consider the
$p$-adic log-Coulomb gases as exact solvable models.

We establish the existence of phase transitions by showing the existence of a
convergence interval $\left(  0,\beta_{UV}\right)  $ for the integral
$\mathcal{Z}_{G,\mathbb{K},\Phi,\mathbf{e}}\left(  \beta\right)  $, such that
the meromorphic continuation of $\mathcal{Z}_{G,\mathbb{K},\Phi,\mathbf{e}%
}\left(  \beta\right)  $ has a pole at $\beta=\beta_{UV}$. We provide two
different types of criteria for the existence of such intervals. The first
type is specific for the $p$-adic case and requires that $\Phi$\ be the
characteristic function of the unit ball $\mathbb{Z}_{p}^{\left\vert
V(G)\right\vert }$, but this criterion works with arbitrary charge
distributions. Second type of criteria works on any local field of
characteristic zero, but it requires that the support of $\Phi$\ be
sufficiently small, and that the charge distribution \ be such that\ in
(\ref{Hamiltonian}) $e_{u}e_{v}=\pm1$ for any $u,v\in V\left(  G\right)  $. In
terms of phase transitions, the log-Coulomb gases studied here behave
similarly to the classical Ising model.

The above mentioned results were established by using the techniques developed
in (\cite{Veys-Zuniga}). In the $p$-adic setting, in the case in which
$\Phi\left(  \boldsymbol{x}\right)  $ is the characteristic function of the
$\left\vert V\left(  G\right)  \right\vert $-dimensional unit ball, the
corresponding partition functions (or local zeta functions) are rational
functions that can be computed explicitly using combinatorial techniques.

A $p$-adic number is a series of the form%
\begin{equation}
x=x_{-k}p^{-k}+x_{-k+1}p^{-k+1}+\ldots+x_{0}+x_{1}p+\ldots,\text{ with }%
x_{-k}\neq0\text{,} \label{p-adic-number}%
\end{equation}
where $p$ denotes a fixed prime number, and the $x_{j}$s \ are $p$-adic
digits, i.e. numbers in the set $\left\{  0,1,\ldots,p-1\right\}  $. There are
natural field operations, sum and multiplication, on series of the form
(\ref{p-adic-number}). The set of all possible $p$-adic numbers constitutes
the field of $p$-adic numbers $\mathbb{Q}_{p}$. There is also a natural norm
in $\mathbb{Q}_{p}$ defined as $\left\vert x\right\vert _{p}=p^{k}$, for a
nonzero $p$-adic number of the form (\ref{p-adic-number}). We extend the
$p$-adic norm to $\mathbb{Q}_{p}^{N}$, by taking $\left\Vert \left(
x_{1},\ldots,x_{N}\right)  \right\Vert _{p}=\max_{i}\left\vert x_{i}%
\right\vert _{p}$.

The Hamiltonian of the $N$-dimensional $p$-adic Coulomb gas is%
\[
H_{N}\left(  x_{1},\ldots,x_{N};\beta\right)  =\sum\limits_{1\leq i<j\leq
N}e_{i}e_{j}E_{\alpha}\left(  \left\Vert x_{i}-x_{j}\right\Vert _{p}\right)
+\frac{1}{\beta}P\left(  x_{1},\ldots,x_{N}\right)  ,
\]
where $e_{j}$ is the charge of a particle located at $x_{j}\in\mathbb{Q}%
_{p}^{N}$, and $P\left(  x_{1},\ldots,x_{N}\right)  $\ is a confining
potential. We assume that $P\left(  x_{1},\ldots,x_{N}\right)  =+\infty$
outside of an open compact subset. The Coulomb kernel $E_{\alpha}(\left\Vert
x\right\Vert _{p})$ is a fundamental solution of a `$p$-adic Poisson's
equation.' More precisely, if%
\[
E_{\alpha}(\left\Vert x\right\Vert _{p})=%
\begin{cases}
\dfrac{1-p^{-\alpha}}{1-p^{\alpha-N}}||x||_{p}^{\alpha-N}, & \text{if }%
\alpha\neq N\\
\dfrac{1-p^{N}}{p^{N}\ln p}\ln||x||_{p}, & \text{if }\alpha=N,
\end{cases}
\]
then $\boldsymbol{D}^{\alpha}E_{\alpha}=\delta$, where $\boldsymbol{D}%
^{\alpha}$, $\alpha>0$, is the $N$-dimensional Taibleson operator which is a
pseudodifferential operator defined as $\mathcal{F}\left(  \boldsymbol{D}%
^{\alpha}\varphi\right)  =||\xi||_{p}^{\alpha}\mathcal{F}\varphi$, where
$\mathcal{F}$ denotes the Fourier transform, see\ \cite[Theorem 13]%
{Rodriguez-Zuniga-2010} and \cite[Chapter 5]{Zuniga-LNM-2016}. The study of
$p$-adic Coulomb gases was initiated in \cite{Zuniga-Torba}, where some
probabilistic aspects attached to Coulomb gases, involving the kernel
$||x||_{p}^{\alpha-N}$,\ $N>\alpha$, were studied.

In this article we study $1$D $p$-adic log-Coulomb gases, under the assumption
that $e^{\frac{-1}{\beta}P}$ is the characteristic function of the $\left\vert
V(G)\right\vert $-dimensional unit ball $\mathbb{Z}_{p}^{\left\vert
V(G)\right\vert }$. In this case, the local zeta function attached to $G$ is
defined as
\[
Z(\boldsymbol{s};G)=\int\limits_{\mathbb{Z}_{p}^{\left\vert V\left(  G\right)
\right\vert }}\prod\limits_{\substack{u,v\in V(G)\\u\sim v}}\left\vert
x_{u}-x_{v}\right\vert _{p}^{s\left(  u,v\right)  }\prod\limits_{v\in
V(G)}dx_{v}\text{,}%
\]
where $s\left(  u,v\right)  $ is a complex variable attached to the edge
connecting the vertices $u$ and $v$. The partition function $\mathcal{Z}%
_{G,p,\mathbf{e}}\left(  \beta\right)  $\ of $H_{p}(\boldsymbol{x}%
;\mathbf{e},\beta,G)$ is related to the local zeta function of the graph by
$\mathcal{Z}_{G,p,\mathbf{e}}\left(  \beta\right)  =\left.  Z(\boldsymbol{s}%
;G)\right\vert _{s\left(  u,v\right)  =e_{u}e_{v}\beta}$.

Section \ref{Section_3} is dedicated to the study of the function
$Z(\boldsymbol{s};G)$. This function admits a meromorphic continuation as a
rational function in the variables $p^{-s\left(  u,v\right)  }$, see
Proposition \ref{Prop1}. We provide a recursive algorithm for computing
$Z(\boldsymbol{s};G)$. The algorithm uses vertex colorings and chromatic
polynomials, see Proposition \ref{Prop2}. This algorithm allows us to describe
the possible poles of $Z(\boldsymbol{s};G)$ in terms of the subgraphs of $G$,
see Theorem \ref{Theorem1} and Corollary \ref{Cor1}.\ 

In Section \ref{Section_4}, we give conditions on the distribution of charges
that guarantee the convergence of the integral $\mathcal{Z}_{G,p,\mathbf{e}%
}\left(  \beta\right)  $ in an interval $\left(  \beta_{IR},\beta_{UV}\right)
$, see Proposition \ref{Prop3}. We also give conditions so that the
meromorphic continuation of $\mathcal{Z}_{G,p,\mathbf{e}}\left(  \beta\right)
$ has a pole at $\beta=\beta_{UV}$, see Proposition \ref{Theorem4}. This
result allows us to give criteria \ for the existence of phase transitions at
finite temperature. In Section \ref{Section_5A}, we study the thermodynamic
limit\ for a log-Coulomb gas attached to a star graph $S_{M}$, confined in the
$1$-dimensional ball $B_{k}$ of radius $p^{k}$, when $M\rightarrow\infty$,
$k\rightarrow\infty$, and $\frac{M}{p^{k}}=\rho$ is constant. Assuming
a\ neutral charge distribution satisfying $e_{v}=\pm1$ for any $v\in V(G)$, we
show that the dimensionless free energy per particle $\beta\mathfrak{f}$\ has
a singularity at $\beta=1$, i.e. the gas has a phase transition at temperature
$k_{B}$. We also compute the grand-canonical partition function for this gas.

There exists a large family of zeta functions attached to finite graphs, which
can be considered as discrete analogues of the Riemann zeta function, see
\cite{Terras} and the references therein. There are also zeta functions
attached to infinite graphs, see e.g. \cite{Clair et al}, \cite{Grigorchuk et
al}, \cite{Guido et al}, and attached to hypergraphs \cite{Kang et al}. From
this perspective our graph zeta function is a `new' mathematical object. On
the other hand,\ our graph zeta functions are related to $p$-adic Feynman
integrals. These integrals were studied by Lerner and Missarov in the context
of quantum field theory, \cite{Lerner-Missarov}, \cite{Lerner-1}, see also
\cite{Dorjevic et al}, \cite{D-K-K-V}, \cite{Parisi}, and the references
therein. In \cite[Theorem 1]{Lerner-1}, under a condition on all the connected
subgraphs of $G$, it was established the convergence of $Z(\boldsymbol{s};G)$,
and a recursive formula was given. Our Theorem \ref{Theorem1} does not require
these conditions.

The connections between zeta functions of number fields and statistical
mechanics, especially phase transitions, have received great attention due to
the influence of the work of Connes, see e.g. \cite{Connes-1}%
-\cite{Connes-Marcolli}, see also \cite{Knauf}. To the best of our knowledge,
the connection between phase transitions and local zeta functions is new. In
\cite{Sinclair} some aspects of the partition function for $p$-adic
log-Coulomb gases attached to the complete graph were studied.

In Section \ref{Section_5} we review the basic aspects of the theory of local
zeta functions for rational functions, on local fields of characteristic zero,
developed in \cite{Veys-Zuniga}. By using this theory, we give a criterion for
the existence of phase transitions at finite temperature for a $1$D
log-Coulomb gas with Hamiltonian (\ref{Hamiltonian}), under the supposition
that the function $\Phi$ is supported on\ a sufficiently small neighborhood of
a point, and that the charge distribution $\boldsymbol{e}=\left\{
e_{v}\right\}  _{v\in V(G)}$ satisfies $\left\{  e_{v}e_{u};v,u\in V(G)\text{,
}u\sim v\right\}  =\left\{  +1,-1\right\}  $, \ see Theorem \ref{Theorem5}.

\section{\label{Section_2}Basic ideas on $p$-adic analysis}

In this section we collect some basic results about $p$-adic analysis that
will be used in the article. For an in-depth review of the $p$-adic analysis
the reader may consult \cite{A-K-S}, \cite{Taibleson}, \cite{V-V-Z}.

\subsection{The field of $p$-adic numbers}

Along this article $p$ will denote a prime number. The field of $p-$adic
numbers $\mathbb{Q}_{p}$ is defined as the completion of the field of rational
numbers $\mathbb{Q}$ with respect to the $p-$adic norm $|\cdot|_{p}$, which is
defined as
\[
\left\vert x\right\vert _{p}=
\begin{cases}
0, & \text{if } x=0\\
p^{-\gamma}, & \text{if } x=p^{\gamma}\frac{a}{b},
\end{cases}
\]
where $a$ and $b$ are integers coprime with $p$. The integer $\gamma=:
\operatorname{ord}(x) $, with $\operatorname{ord}(0):=+\infty$, is called
the\textit{\ }$p-$\textit{adic order of} $x$.

Any $p-$adic number $x\neq0$ has the form $x=p^{\operatorname{ord}(x)}%
\sum_{j=0}^{\infty}x_{j}p^{j}$, where $x_{j}\in\{0,\dots,p-1\}$ and $x_{0}%
\neq0$.

\subsection{Topology of $\mathbb{Q}_{p}^{N}$}

We extend the $p-$adic norm to $\mathbb{Q}_{p}^{N}$ by taking
\[
||x||_{p}:=\max_{1\leq i\leq N}|x_{i}|_{p},\qquad\text{for }x=(x_{1}%
,\dots,x_{N})\in\mathbb{Q}_{p}^{N}.
\]
We define $\operatorname{ord}(x)=\min_{1\leq i\leq N}\{\operatorname{ord}%
(x_{i})\}$, then $||x||_{p}=p^{-\operatorname{ord}(x)}$. The metric space
$\left(  \mathbb{Q}_{p}^{N},||\cdot||_{p}\right)  $ is a separable complete
ultrametric space. Ultrametricity refers to the fact that the norm
$||\cdot||_{p}$ satisfies $||x+y||_{p}\leq\max\left\{  ||x||_{p}%
,||y||_{p}\right\}  $. Furthermore, if $||x||_{p}\neq||y||_{p}$, then
$||x+y||_{p}=\max\left\{  ||x||_{p},||y||_{p}\right\}  $.

For $r\in\mathbb{Z}$, denote by $B_{r}^{N}(a)=\{x\in\mathbb{Q}_{p}%
^{N};||x-a||_{p}\leq p^{r}\}$ \textit{the ball of radius }$p^{r}$ \textit{with
center at} $a=(a_{1},\dots,a_{N})\in\mathbb{Q}_{p}^{N}$, and take $B_{r}%
^{N}:=B_{r}^{N}(0)$. Note that $B_{r}^{N}(a)=B_{r}(a_{1})\times\cdots\times
B_{r}(a_{N})$, where $B_{r}(a_{i}):=\{x_{i}\in\mathbb{Q}_{p};|x_{i}-a_{i}%
|_{p}\leq p^{r}\}$ is the one-dimensional ball of radius $p^{r}$ with center
at $a_{i}\in\mathbb{Q}_{p}$. The ball $B_{0}^{N}$ equals to the product of $N$
copies of $B_{0}=\mathbb{Z}_{p}$, \textit{the ring of }$p-$\textit{adic
integers of }$\mathbb{Q}_{p}$.

\subsection{Test functions}

A complex-valued function $\varphi$ defined on $\mathbb{Q}_{p}^{N}$ is called
\textit{locally constant} if for any $x\in\mathbb{Q}_{p}^{N}$ there exist an
integer $l(x)\in\mathbb{Z}$ such that $\varphi(x+x^{\prime})=\varphi(x)\qquad
$for $x^{\prime}\in B_{l(x)}^{N}$. A function $\varphi:\mathbb{Q}_{p}%
^{N}\rightarrow\mathbb{C}$ is called a \textit{Bruhat-Schwartz function,} or a
\textit{test function,} if it is locally constant with compact support. The
$\mathbb{C}$-vector space of Bruhat-Schwartz functions is denoted by
$\mathcal{D}:=\mathcal{D}(\mathbb{Q}_{p}^{N})$.

\subsection{Integration and change of variables}

We denote by $d^{N}x$ a Haar measure of the topological group $(\mathbb{Q}%
_{p}^{N},+)$ normalized by the condition $\int_{B_{0}^{N}}d^{N}x=1$.

A function $h:U\rightarrow\mathbb{Q}_{p}$ is said to be \textit{analytic} on
an open subset $U\subset\mathbb{Q}_{p}^{N}$, if there exists a convergent
power series $\sum_{i}a_{i}x^{i}$ for $x\in\widetilde{U}\subset U$ , with
$\widetilde{U}$ open, such that $h\left(  x\right)  =\sum_{i}a_{i}x^{i}$ for
$x\in\widetilde{U}$, with $x^{i}=x_{1}^{i_{1}}\cdots x_{N}^{i_{N}}$,
$i=\left(  i_{1},\ldots,i_{N}\right)  $. In this case, $\frac{\partial
}{\partial x_{l}}h\left(  x\right)  =\sum_{i}a_{i}\frac{\partial}{\partial
x_{l}}\left(  x^{i}\right)  $ is a convergent power series. Let $U$, $V$ be
open subsets of $\mathbb{Q}_{p}^{N}$. A mapping $\sigma:U\rightarrow V$,
$\sigma=\left(  \sigma_{1},\ldots,\sigma_{N}\right)  $ is called analytic if
each $\sigma_{i}$ is analytic.

Let $\varphi:V$ $\rightarrow\mathbb{C}$ be a continuous function with compact
support, and let $\sigma:U\rightarrow V$ $\ $be an analytic mapping. Then
\begin{equation}%
%TCIMACRO{\tint \limits_{V}}%
%BeginExpansion
{\textstyle\int\limits_{V}}
%EndExpansion
\varphi\left(  y\right)  d^{N}y=%
%TCIMACRO{\tint \limits_{U}}%
%BeginExpansion
{\textstyle\int\limits_{U}}
%EndExpansion
\varphi\left(  \sigma(x)\right)  \left\vert Jac(\sigma(x))\right\vert
_{p}d^{N}x\text{,} \label{Change-of_var}%
\end{equation}
where $Jac(\sigma(z))=\det\left[  \frac{\partial\sigma_{i}}{\partial x_{j}%
}\left(  z\right)  \right]  _{\substack{1\leq i\leq N\\1\leq j\leq N}}$, see
e.g. \cite[Section 10.1.2]{Bourbaki}.

\section{\label{Section_3}Zeta functions for graphs}

Along this article by a graph, we mean a finite, simple graph, i.e. a\ graph
with no loops and no multiple edges, see e.g. \cite[Definition 1.2.4]%
{Balakrishnan et al}.

Let $G$ be a graph. We denote by $V:=V(G)$ its set of vertices and by
$E:=E(G)$ its set of edges. If $E(G)\neq\emptyset$, we denote by $i_{G}$ the
incidence relation on $G$, i.e. a mapping from the set of edges to the set of
pairs of vertices, where the corresponding two vertices are necessarily
distinct. We use the notation $i_{G}\left(  l\right)  =\left\{  u,v\right\}  $
or the notation $u\sim v$. To each vertex $v\in V$ we attach a $p$-adic
variable $x_{v}$, and to each edge $l\in E$ we attach a complex variable
$s\left(  l\right)  $. We also use the notation $s\left(  u,v\right)  $ if
$u\sim v$. We set $\boldsymbol{x}:=\left\{  x_{v}\right\}  _{v\in V}$,
$\boldsymbol{s}:=\left\{  s\left(  l\right)  \right\}  _{l\in E}$.

Given $l\in E$, with $i_{G}\left(  l\right)  =\left\{  u,v\right\}  $, we set
\[
F_{l}\left(  x_{u},x_{v},s\left(  l\right)  \right)  :=\left\vert x_{u}%
-x_{v}\right\vert _{p}^{s\left(  l\right)  }%
\]
and
\begin{equation}
F_{G}\left(  \boldsymbol{x},\boldsymbol{s}\right)  :=\prod\limits_{l\in
E}F_{l}\left(  x_{u},x_{v},s\left(  l\right)  \right)  =\prod
\limits_{\substack{u,v\in V\\u\sim v}}\left\vert x_{u}-x_{v}\right\vert
_{p}^{s\left(  u,v\right)  }. \label{Function_F_G}%
\end{equation}

\begin{remark}
\label{Nota0}(i) If $V(G)\neq\emptyset$ and $E(G)=\emptyset$, then $G$
consists of a finite set of vertices without edges connecting them, thus
incidence relation is not defined. In this case we set $F_{G}\left(
\boldsymbol{x},\boldsymbol{s}\right)  :=1$. Due to technical reasons, we
consider the empty set as a graph, in this case $F_{\varnothing}\left(
\boldsymbol{x},\boldsymbol{s}\right)  :=1$.
\end{remark}

\begin{notation}
(i) For a finite subset $A$, we denote by $\left\vert A\right\vert $ its cardinality.

\noindent(ii) We denote by $\mathcal{D}_{sym}\left(  \mathbb{Q}_{p}%
^{N}\right)  $ the $\mathbb{C}$-vector space of symmetric test functions, i.e.
all the complex-valued test functions satisfying $\varphi\left(  x_{1}%
,\ldots,x_{N}\right)  =\varphi\left(  x_{\pi\left(  1\right)  },\ldots
,x_{\pi\left(  N\right)  }\right)  $ for any permutation $\pi$ of $\left\{
1,2,\ldots,N\right\}  $.
\end{notation}

Let $G$ and $H$ be graphs. By a graph isomorphism $\sigma:G\rightarrow H$, we
mean a pair of mappings $\left\{  \sigma_{E},\sigma_{V}\right\}  $, where
$\sigma_{V}:V(G)\rightarrow V(H)$, $\sigma_{E}:E(G)\rightarrow E(H)$ are
bijections, with the property that $i_{G}\left(  l\right)  =\left\{
u,v\right\}  $ if and only if $i_{H}\left(  \sigma_{E}\left(  l\right)
\right)  =\left\{  \sigma_{V}\left(  u\right)  ,\sigma_{V}\left(  v\right)
\right\}  $. In the case of simple graphs, $\sigma_{E}$ is completely
determined by $\sigma_{V}$. For the sake of simplicity, we will denote the
pair $\left\{  \sigma_{E},\sigma_{V}\right\}  $ as $\sigma$, see e.g.
\cite[Sections 1.2.9, 1.2.10]{Balakrishnan et al}.

We denote by Aut$\left(  G\right)  $ the automorphism group of $G$. Let
$\sigma:G\rightarrow H$ be a graph isomorphism. Assume that the cardinality of
$\left\vert V(G)\right\vert =\left\vert V(H)\right\vert =N$.\ Let $x_{u}$,
$u\in V(G)$, be \ $p$-adic variables as before. Then the mapping%
\begin{equation}%
\begin{array}
[c]{llll}%
\sigma^{\ast}: & \mathbb{Q}_{p}^{N} & \rightarrow & \mathbb{Q}_{p}^{N}\\
& x_{v} & \rightarrow & x_{\sigma\left(  v\right)  }%
\end{array}
\label{Sigma_change_of_var}%
\end{equation}
is a $p$-adic analytic isomorphism that preserves the Haar measure of
$\mathbb{Q}_{p}^{N}$, see (\ref{Change-of_var}).

\begin{definition}
Given $\varphi\in\mathcal{D}_{sym}\left(  \mathbb{Q}_{p}^{\left\vert V\left(
G\right)  \right\vert }\right)  $, the $p$-adic zeta function attached to
$\left(  G,\varphi\right)  $ is defined \ as%
\[
Z_{\varphi}(\boldsymbol{s};G)=\int\limits_{\mathbb{Q}_{p}^{\left\vert V\left(
G\right)  \right\vert }}\varphi\left(  \boldsymbol{x}\right)  F_{G}\left(
\boldsymbol{x},\boldsymbol{s}\right)  \prod\limits_{v\in V(G)}dx_{v}\text{,}%
\]
for $\operatorname{Re}(s\left(  l\right)  )>0$ for every $l\in E$, where $%
%TCIMACRO{\tprod \nolimits_{v\in V\left(  G\right)  }}%
%BeginExpansion
{\textstyle\prod\nolimits_{v\in V\left(  G\right)  }}
%EndExpansion
dx_{v}$ denotes the normalized Haar measure on $\left(  \mathbb{Q}%
_{p}^{\left\vert V\left(  G\right)  \right\vert },+\right)  $. If $\varphi$ is
the characteristic function of $\mathbb{Z}_{p}^{\left\vert V\left(  G\right)
\right\vert }$, we use the notation $Z(\boldsymbol{s};G)$.
\end{definition}

\begin{lemma}
\label{Lemma1}Let $G$ and $H$ be graphs. If $\sigma:G\rightarrow H$ is a graph
isomorphism, then
\[
Z_{\varphi}(\left\{  s\left(  l\right)  \right\}  _{l\in E(G)};G)=Z_{\varphi
}(\left\{  s\left(  l\right)  \right\}  _{l\in E(H)};H).
\]
Furthermore, for any $\sigma=\left(  \sigma_{V},\sigma_{E}\right)  \in
$Aut$\left(  G\right)  $, it holds true that
\begin{equation}
Z(\left\{  s\left(  l\right)  \right\}  _{l\in E(G)};G)=Z(\left\{  s\left(
\sigma_{E}\left(  l\right)  \right)  \right\}  _{l\in E(G)};G),
\label{Eq_symmetry}%
\end{equation}
where the integrals exist.
\end{lemma}

\begin{proof}
By using that
\[
Z_{\varphi}(\boldsymbol{s};G)=\int\limits_{\mathbb{Q}_{p}^{\left\vert V\left(
G\right)  \right\vert }}\varphi\left(  \left\{  x_{v}\right\}  _{v\in
V(G)}\right)  \prod\limits_{\substack{u,v\in V(G)\\u\sim v}}\left\vert
x_{u}-x_{v}\right\vert _{p}^{s\left(  u,v\right)  }\prod\limits_{v\in
V(G)}dx_{v}\text{,}%
\]
and changing variables as $\sigma^{\ast}:$ $\mathbb{Q}_{p}^{N}\rightarrow
\mathbb{Q}_{p}^{N}$, $x_{v}\mapsto x_{\sigma(v)}$, see
(\ref{Sigma_change_of_var}), we have
\[
\varphi\left(  \left\{  x_{v}\right\}  _{v\in V(G)}\right)  =\varphi\left(
\left\{  x_{\sigma\left(  v\right)  }\right\}  _{v\in V(G)}\right)
=\varphi\left(  \left\{  x_{v^{\prime}}\right\}  _{v^{\prime}\in V(H)}\right)
,
\]
because the list $\left\{  x_{v^{\prime}}\right\}  _{v^{\prime}\in V(H)}$ is a
permutation of the list $\left\{  x_{v}\right\}  _{v\in V(G)}$. In addition,
\begin{align*}
\prod\limits_{\substack{u,v\in V(G)\\u\sim v}}\left\vert x_{u}-x_{v}%
\right\vert _{p}^{s\left(  u,v\right)  }  &  =\prod\limits_{\substack{\sigma
\left(  u\right)  ,\sigma\left(  v\right)  \\u,v\in V(G)\\u\sim v}}\left\vert
x_{\sigma\left(  u\right)  }-x_{\sigma\left(  v\right)  }\right\vert
_{p}^{s\left(  \sigma\left(  u\right)  ,\sigma\left(  v\right)  \right)  }\\
&  =\prod\limits_{\substack{u^{\prime},v^{\prime}\in V(H)\\u^{\prime}\sim
v^{\prime}}}\left\vert x_{u^{\prime}}-x_{v^{\prime}}\right\vert _{p}%
^{s(u^{\prime},v^{\prime})},
\end{align*}
and by using that $\sigma^{\ast}$ preserves the Haar measure,%
\[
\prod\limits_{v\in V(G)}dx_{v}=\prod\limits_{v\in V(G)}dx_{\sigma\left(
v\right)  }=\prod\limits_{v^{\prime}\in V(H)}dx_{v^{\prime}}.
\]
Consequently $Z_{\varphi}(\left\{  s\left(  l\right)  \right\}  _{l\in
E(G)};G)=Z_{\varphi}(\left\{  s\left(  l\right)  \right\}  _{l\in E(H)};H)$.
\end{proof}

\begin{remark}
We use the notation $G=G_{1}\#\cdots\#G_{k}$ to mean that $G_{1},\cdots,G_{k}$
are all the distinct connected components of $G$. Then $F_{G}\left(
\boldsymbol{x},\boldsymbol{s}\right)  =%
%TCIMACRO{\tprod \nolimits_{i=1}^{k}}%
%BeginExpansion
{\textstyle\prod\nolimits_{i=1}^{k}}
%EndExpansion
F_{G_{i}}\left(  \boldsymbol{x},\boldsymbol{s}\right)  $ and
\[
Z(\boldsymbol{s};G)=%
%TCIMACRO{\tprod \nolimits_{i=1}^{k}}%
%BeginExpansion
{\textstyle\prod\nolimits_{i=1}^{k}}
%EndExpansion
Z(\boldsymbol{s};G_{i}).
\]
Notice that $Z(\boldsymbol{s};G_{i})=1$, if $G_{i}$ consists of only one vertex.
\end{remark}

The zeta functions $Z_{\varphi}(\boldsymbol{s};G)$ are a special type of
multivariate Igusa zeta functions. These functions were studied in
\cite{Loeser}, in particular, the following result holds true:

\begin{proposition}
[{F. Loeser \cite[Th\'{e}or\`{e}me 1.1.4]{Loeser}}]\label{Prop1}The zeta
function $Z_{\varphi}(\boldsymbol{s};G)$\ admits a meromorphic continuation to
$\mathbb{C}^{\left\vert E\left(  G\right)  \right\vert }$ as a rational
function in the variables $p^{-s\left(  l\right)  }$, $l\in E\left(  G\right)
$, more precisely,%
\begin{equation}
Z_{\varphi}(\boldsymbol{s};G)=\frac{P_{\varphi}(\boldsymbol{s})}%
{\prod\limits_{i\in T}\left(  1-p^{-N_{0}^{i}-\sum_{l\in E\left(  G\right)
}N_{l}^{i}s(l)}\right)  }, \label{Formula_1}%
\end{equation}
where $T$ is a finite set, the $N_{0}^{i}$, $N_{l}^{i}$ are non-negative
integers, and $P_{\varphi}(\boldsymbol{s})$ is a polynomial in the variables
$\left\{  p^{-s\left(  l\right)  }\right\}  _{l\in E\left(  G\right)  }$.
\end{proposition}

\begin{corollary}
\label{Cor1_A}The following functional equations hold true:%
\[
\frac{P_{\varphi}(\left\{  s\left(  l\right)  \right\}  _{l\in E(G)})}%
{\prod\limits_{i\in T}\left(  1-p^{-N_{0}^{i}-\sum_{l\in E\left(  G\right)
}N_{l}^{i}s(l)}\right)  }=\frac{P_{\varphi}(\left\{  s\left(  \sigma
_{E}\left(  l\right)  \right)  \right\}  _{l\in E(G)})}{\prod\limits_{i\in
T}\left(  1-p^{-N_{0}^{i}-\sum_{\sigma_{E}\left(  l\right)  \in E\left(
G\right)  }N_{\sigma_{E}\left(  l\right)  }^{i}s(\sigma_{E}\left(  l\right)
)}\right)  },
\]
for any $\sigma=\left(  \sigma_{V},\sigma_{E}\right)  \in$Aut$\left(
G\right)  $.
\end{corollary}

\begin{proof}
The results follows from (\ref{Eq_symmetry}) by using the fact
(\ref{Formula_1}) gives an equality between functions in an open set
containing $\left\{  \operatorname{Re}(s(l))>0;l\in E(G)\right\}  $.
\end{proof}

\begin{example}
Let $K_{2}$ be the complete graph with two vertices, $v_{0}$, $v_{1}$. We
denote the corresponding edge as $l$. Then $F_{K_{2}}\left(  \boldsymbol{x}%
,\boldsymbol{s}\right)  =\left\vert x_{v_{0}}-x_{v_{1}}\right\vert
_{p}^{s\left(  l\right)  }$ and
\[
Z(\boldsymbol{s};K_{2})=\int_{\mathbb{Z}_{p}^{2}}\left\vert x_{v_{0}}%
-x_{v_{1}}\right\vert _{p}^{s\left(  l\right)  }dx_{v_{0}}dx_{v_{1}}%
=\int_{\mathbb{Z}_{p}}\left\{  \int_{\mathbb{Z}_{p}}\left\vert x_{v_{0}%
}-x_{v_{1}}\right\vert _{p}^{s\left(  l\right)  }dx_{v_{0}}\right\}
dx_{v_{1}}.
\]
By changing variables as $y=x_{v_{0}}-x_{v_{1}}$, $z=x_{v_{1}}$, we have
\[
Z(\boldsymbol{s};K_{2})=\int_{\mathbb{Z}_{p}}\left\{  \int_{\mathbb{Z}_{p}%
}\left\vert y\right\vert _{p}^{s\left(  l\right)  }dy\right\}  dz=\int
_{\mathbb{Z}_{p}}\left\vert y\right\vert _{p}^{s\left(  l\right)  }%
dy=\frac{1-p^{-1}}{1-p^{-1-s(l)}}.
\]

\end{example}

\begin{example}
\label{Example_Star}We denote by $S_{N}$ the star graph with $N$ vertices
labeled as $V(S_{N})=\left\{  1,\ldots,N\right\}  $, where the vertex $1$ is
the center of the star, i.e.
\[
E(S_{N})=\left\{  \left\{  1,2\right\}  ,\cdots,\left\{  1,l\right\}
,\cdots,\left\{  1,N\right\}  \right\}  .
\]
Then $F_{S_{N}}\left(  \boldsymbol{x},\boldsymbol{s}\right)  =\prod
\limits_{i=2}^{N}\left\vert x_{1}-x_{i}\right\vert _{p}^{s_{i}}$ and%
\[
Z(\boldsymbol{s};S_{N})=\int_{\mathbb{Z}_{p}}\left\{  \int_{\mathbb{Z}%
_{p}^{N-1}}\prod\limits_{i=2}^{N}\left\vert x_{1}-x_{i}\right\vert _{p}%
^{s_{i}}\prod\limits_{i=2}^{N}dx_{i}\right\}  dx_{1}.
\]
By changing variables as $z_{1}=x_{1}$, $z_{i}=x_{1}-x_{i}$ for $i=2,\ldots
,N$, we obtain that%
\[
Z(\boldsymbol{s};S_{N})=\int_{\mathbb{Z}_{p}^{N-1}}\prod\limits_{i=2}%
^{N}\left\vert z_{i}\right\vert _{p}^{s_{i}}\prod\limits_{i=2}^{N}dz_{i}%
=\prod\limits_{i=2}^{N}\int_{\mathbb{Z}_{p}}\left\vert z_{i}\right\vert
_{p}^{s_{i}}dz_{i}=\frac{\left(  1-p^{-1}\right)  ^{N-1}}{\prod\limits_{i=2}%
^{N}\left(  1-p^{-1-s_{i}}\right)  }.
\]

\end{example}

\begin{example}
\label{Example_Tree}Let $T_{N}$ be a finite connected tree with $N$ vertices.
Then
\[
Z(\boldsymbol{s},T_{N})=\frac{(1-p^{-1})^{N-1}}{\prod_{\{u,v\}\in E(T_{N}%
)}1-p^{-1-s(u,v)}}.
\]

We recall that a tree is an undirected graph in which any two vertices are
connected by exactly one path. We fixed $r\in V(T_{N})$ and for $r\in
V(T_{N})$ we denote by $l_{r}(v)$ the length of path from $r$ to $v$. We now
set $l_{r}(T_{N}):=\max_{r\in V(T)}l_{r}(v)$. If $l_{r}(T_{N})=1$, then
$T_{N}$ is a star graph with $N$ vertices. The announced formula is establihed
by induction on $l_{r}(T_{N})$. The case $l_{r}(T_{N})=1$ was already
established. Assume that $l_{r}(T_{N})\geq2$. Then there exists $u^{\prime}\in
V(T_{N})\setminus\{r\}$ with $l_{r}(u^{\prime})=l_{r}(T_{N})$. We fix a such
$u^{\prime}$, then there exists a unique path from $u^{\prime}$ to $r$,
and\ consequently a unique $v^{\prime}\in V(T_{N})$ with $u^{\prime}\sim
v^{\prime}$. We denote by $T_{N-1}^{\prime}$ the tree obtained from $T_{N}$ by
deleting the edge $u^{\prime}\sim v^{\prime}$. Notice that $T_{N-1}^{\prime}$
has $N-1$ vertices. Then%
\begin{multline*}
Z(\boldsymbol{s},T_{N})=\\%
%TCIMACRO{\dint \limits_{\mathbb{Z}_{p}^{|V(T_{N})|}}}%
%BeginExpansion
{\displaystyle\int\limits_{\mathbb{Z}_{p}^{|V(T_{N})|}}}
%EndExpansion
\left(  \prod_{\substack{u,v\in\left(  V(T_{N-1}^{\prime})\smallsetminus
\left\{  u^{\prime}\right\}  \right)  \\u\sim v}}|x_{u}-x_{v}|_{p}%
^{s(u,v)}\right)  |x_{u^{\prime}}-x_{v^{\prime}}|_{p}^{s(u^{\prime},v^{\prime
})}dx_{u^{\prime}}\prod_{v\in\left(  V(T_{N-1}^{\prime})\smallsetminus\left\{
u^{\prime}\right\}  \right)  }dx_{v}.
\end{multline*}
We now change variables as $x_{u}\mapsto x_{u}$ if $u\neq u^{\prime}$, and
$x_{u}\mapsto z_{u}+x_{v^{\prime}}$ if $u\neq u^{\prime}$, in the above
integral:
\begin{multline*}
Z(\boldsymbol{s},T_{N})=\\%
%TCIMACRO{\dint \limits_{\mathbb{Z}_{p}^{|V(T_{N-1}^{\prime})|}}}%
%BeginExpansion
{\displaystyle\int\limits_{\mathbb{Z}_{p}^{|V(T_{N-1}^{\prime})|}}}
%EndExpansion
\text{ }%
%TCIMACRO{\dint \limits_{\mathbb{Z}_{p}}}%
%BeginExpansion
{\displaystyle\int\limits_{\mathbb{Z}_{p}}}
%EndExpansion
\left(  \prod_{\substack{u,v\in\left(  V(T_{N-1}^{\prime})\smallsetminus
\left\{  u^{\prime}\right\}  \right)  \\u\sim v}}|x_{u}-x_{v}|_{p}%
^{s(u,v)}\right)  |z_{u}|_{p}^{s(u^{\prime},v^{\prime})}dz_{u}\prod
_{v\in\left(  V(T_{N-1}^{\prime})\smallsetminus\left\{  u^{\prime}\right\}
\right)  }dx_{v}\\
=Z(\boldsymbol{s},T_{N-1}^{\prime})%
%TCIMACRO{\dint \limits_{\mathbb{Z}_{p}}}%
%BeginExpansion
{\displaystyle\int\limits_{\mathbb{Z}_{p}}}
%EndExpansion
|z_{u}|_{p}^{s(u^{\prime},v^{\prime})}dz_{u}=Z(\boldsymbol{s},T_{N-1}^{\prime
})\left(  \frac{1-p^{-1}}{1-p^{-1-s(u^{\prime},v^{\prime})}}\right)  .
\end{multline*}
Thus, by induction hypothesis,
\begin{multline*}
Z(\boldsymbol{s},T_{N})=\left(  \frac{1-p^{-1}}{1-p^{-1-s(u^{\prime}%
,v^{\prime})}}\right)  \left(  \frac{(1-p^{-1})^{|V(T_{N-1}^{\prime})|-1}%
}{\prod_{\{u,v\}\in E(T_{N-1}^{\prime})}1-p^{-1-s(u,v)}}\right) \\
=\frac{(1-p^{-1})^{N}}{\prod_{\{u,v\}\in E(T)}1-p^{-1-s(u,v)}}.
\end{multline*}

\end{example}

\begin{example}
\label{example_3}Let $L_{N}$ denote the linear graph consisting of $N$
vertices labeled as $V(L_{N})=\left\{  1,\ldots,N\right\}  $, and edges
$E(L_{N})=\left\{  \left\{  1,2\right\}  ,\cdots,\left\{  l-1,l\right\}
,\cdots,\left\{  N-1,N\right\}  \right\}  $. Then $F_{L_{N}}\left(
\boldsymbol{x},\boldsymbol{s}\right)  =\prod\limits_{i=2}^{N}\left\vert
x_{i-1}-x_{i}\right\vert _{p}^{s_{i}}$ and%
\[
Z(\boldsymbol{s};L_{N})=\int_{\mathbb{Z}_{p}^{N}}\prod\limits_{i=2}%
^{N}\left\vert x_{i-1}-x_{i}\right\vert _{p}^{s_{i}}\prod\limits_{i=2}%
^{N}dx_{i}.
\]
By changing variables as $z_{1}=x_{1}$, $z_{i}=x_{i-1}-x_{i}$ for
$i=2,\ldots,N$ and using the fact that this transformation preserves the
normalized Haar measure of $\mathbb{Z}_{p}^{N}$, we obtain that%
\[
Z(\boldsymbol{s};L_{N})=\prod\limits_{i=2}^{N}\int_{\mathbb{Z}_{p}}\left\vert
z_{i}\right\vert _{p}^{s_{i}}dz_{i}=\frac{\left(  1-p^{-1}\right)  ^{N-1}%
}{\prod\limits_{i=2}^{N}\left(  1-p^{-1-s_{i}}\right)  }=Z(\boldsymbol{s}%
;S_{N}).
\]

\end{example}

\begin{remark}
The assertion%
\[
\text{if }Z(\boldsymbol{s};G)\neq Z(\boldsymbol{s};K)\text{, then }G\text{ is
not isomorphic to }K\text{ }%
\]
is true, cf. Lemma \ref{Lemma1}, but Examples \ref{Example_Star},
\ref{example_3}\ show that the assertion
\[
\text{if }Z(\boldsymbol{s};G)=Z(\boldsymbol{s};K)\text{, then }G\text{ is
isomorphic to }K\text{ }%
\]
is false.
\end{remark}

\subsection{Vertex Colorings and Chromatic Functions}

We recall that a graph $H$ is called a \textit{subgraph} of $G$ if
$V(H)\subset V(G)$, $E(H)\subset E(G)$. If $E(H)\neq\emptyset$, $i_{H}$ is the
restriction of $i_{G}$ to $E(H)$. If $E(H)=\emptyset$, $H$ consists of a
subset of vertices of $G$ without edges, and thus $i_{H}$ is not defined.

\begin{definition}
\label{Def_Sub_generated}Let $I$ be a non-empty subset of $V(G)$. We denote by
$G_{I}$ (or $G\left[  I\right]  $) the \textit{subgraph induced by} $I$, which
is the subgraph defined as $V(G_{I})=I$,
\[
E(G_{I})=\left\{  l\in E(G);i_{G}\left(  l\right)  =\left\{  v,v^{\prime
}\right\}  \text{ for some }v,v^{\prime}\in I\right\}  ,
\]
and $i_{G_{I}}=i_{G}\mid_{E(G_{I})}$. If $I=\emptyset$, by definition
$G_{I}=\emptyset$.

Suppose that $G_{I}=G_{I}^{(1)}\#\cdots\#G_{I}^{\left(  m\right)  }$. If
$G_{I}^{(j)}=\left\{  v\right\}  $, we say that $v$ is an \textit{isolated
vertex} of $G_{I}$. We denote by $G_{I}^{\text{iso}}$ the set of all the
isolated vertices of $G_{I}$. Then%
\[
G_{I}=G_{I}^{\text{red}}%
%TCIMACRO{\tbigsqcup }%
%BeginExpansion
{\textstyle\bigsqcup}
%EndExpansion
G_{I}^{\text{iso}},
\]
where $G_{I}^{\text{red}}:=G_{I}^{(i_{1})}\#\cdots\#G_{I}^{\left(
i_{l}\right)  }$ and $\left\vert G_{I}^{\left(  i_{k}\right)  }\right\vert >1$
for $k=1,\ldots,l$. We call $G_{I}^{\text{red}}$ the \textit{reduced subgraph
of }$G_{I}$. We adopt the convention that if $I=\varnothing$, then
$G_{I}^{\text{red}}=G_{I}^{\text{iso}}=\varnothing$.
\end{definition}

\subsubsection{Colorings and Chromatic Functions}

In this section we color graphs using $p$ colors, more precisely, we attach to
every element of $\left\{  0,1,\ldots,p-1\right\}  $ (which we identify with
an element of $\mathbb{F}_{p}$) a color.

\begin{definition}
\label{Def_Vertex_coloring}A vertex coloring of $G$ is a mapping
$C:V(G)\rightarrow\mathbb{F}_{p}$. If $v$ is a vertex of $G$, then $C(v)$ is
its color. We denote by $Colors(G)$, the set of all possible vertex-colorings
of $G$.
\end{definition}

Notice that any coloring $C$ is given by a vector $\boldsymbol{a}=\left(
a_{v}\right)  _{v\in V(G)}\in\mathbb{F}_{p}^{\left\vert V\left(  G\right)
\right\vert }$ with $C(v)=a_{v}$ for $v\in V$. We will identify $C$ with
$\boldsymbol{a}$. Our notion of vertex coloring\ is completely different from
the classical one which requires \ that adjacent vertices of $G$ receive
distinct colors of $\mathbb{F}_{p}$, see e.g. \cite[Section 7.2]{Balakrishnan
et al}.

\begin{definition}
\label{Def_colorig_graph}Given a pair $\left(  G,C\right)  $, we attach to it
a colored graph $G^{C}$ defined as follows: $V(G^{C})=V(G)$,
\[
E(G^{C})=\left\{  l\in E(G);C(u)=C(v)\text{ where }i_{G}(l)=\left\{
u,v\right\}  \right\}
\]
and $i_{G^{C}}=i_{G}\mid_{E(G^{C})}$.
\end{definition}

We note that if $G_{1}^{C},\cdots,G_{r}^{C}$, with $r=r(C)$, are all the
connected components of $G^{C}$, then $C\mid_{G_{k}^{C}}$ is constant for
$k=1,\ldots,r$. If $C$ is identified with $\boldsymbol{a}$ we use the notation
$G^{\boldsymbol{a}}$. Definition \ref{Def_colorig_graph} tell us how to color
the edges of a graph if we have already assigned colors to the vertices of the
graph. To an edge having its two vertices colored with the same color we
assign the color of its vertices, in other case, we discard the edge.

\begin{definition}
We set $Colored(G):=\left\{  G^{C};C\in Colors(G)\right\}  $, and
$Subgraphs(G,\left\vert G\right\vert )$ to be the set of all graphs $H$ such
that $V(H)=V(G)$, $E(H)\subset E(G)$, and if $E(H)\neq\emptyset$, $i_{H}$ is
the restriction of $i_{G}$ to $E(H)$. We define
\[
\mathfrak{F}:Colored(G)\rightarrow Subgraph(G,\left\vert G\right\vert )
\]
as follows: $\mathfrak{F}\left(  G^{C}\right)  =H$ if and only if
$V(H)=V(G^{C})$, $E(H)=E(G^{C})$ $\ $and $i_{H}=i_{G^{C}}$. We set
$Subgraph_{\mathfrak{F}}(G,\left\vert G\right\vert )=\mathfrak{F}\left(
Colored(G)\right)  $.
\end{definition}

The family $Colored(G)$ is formed by all the possible colored versions of $G$,
the operation `forgetting the coloring' $\mathfrak{F}$ assigns to an element
of $Colored(G)$ a subgraph of $G$ having the same vertices as $G$. Any graph
in $Subgraphs(G,\left\vert G\right\vert )$ is obtained from $G$ by deleting
one or more edges, `but keeping'\ the corresponding vertices.

\begin{definition}
We define $Indgraphs(G)$ to be the set of all connected graphs $H$ such that
there exists a coloring $C$, with $G^{C}=G_{1}^{C}\#\cdots\#G_{r}^{C}$, and
$H=G_{i}^{C}$ for exactly one index $i$.
\end{definition}

By Definition \ref{Def_Sub_generated}, we have%
\[
Indgraphs(G)=\left\{  G\left[  I\right]  ;\varnothing\neq I\subset V(G)\text{
and }G\left[  I\right]  \text{ is connected}\right\}  ,
\]
where $G\left[  I\right]  $ denotes the subgraph induced by $I$.

\subsubsection{The Chromatic Functions}

\begin{definition}
Given $H$ in $Subgraphs(G,\left\vert G\right\vert )$, we define its chromatic
function as%
\[
\mathcal{C}(p;H)=\left\vert \left\{  G^{C}\in Colored(G);\mathfrak{F}\left(
G^{C}\right)  =H\right\}  \right\vert .
\]

\end{definition}

Notice that if $G$ is connected, then $\mathcal{C}(p;G)=p$. Indeed, if we use
at least two colors then $G^{C}$ has at least two connected components, and
thus $\mathcal{F}(G^{C})\neq G$. So we can use only constant colorings to have
$\mathcal{F}(G^{C})=G$.

Given $u$, $v\in V(G)$, we denote by $d(u,v)$ the length of the shortest path
in $G$ joining $u$ and $v$. Given $H$, $W$ subgraphs of $G$, we set
\[
d(H,W)=\min_{u\in V(H)\text{, }v\in V(W)}d(u,v)\in\mathbb{N}.
\]

\begin{remark}
Suppose that $H=H_{1}\#\cdots\#H_{r}$. The condition $\mathfrak{F}\left(
G^{C}\right)  =H$ implies \ that $\left.  C\right\vert _{H_{i}}=a_{i}%
\in\mathbb{F}_{p}$ for $i=1,\ldots,l$. Now if $d(H_{i},H_{j})=1$, then
$a_{i}\neq a_{j}$, i.e. $a_{i}\neq a_{j}$ if $d(H_{i},H_{j})=1$. If
$d(H_{i},H_{j})\geq2$, the colors $a_{i}$, $a_{j}$ may be equal. We now define%
\[
D_{1}(H):=D_{1}=\left\{  \left\{  H_{i},H_{j}\right\}  ;H_{i},H_{j}\text{ are
connected components of }H\text{, }d(H_{i},H_{j})=1\right\}  ,
\]
and
\[
D_{2}(H):=D_{2}=\left\{  \left\{  H_{i},H_{j}\right\}  ;H_{i},H_{j}\text{ are
connected components of }H\text{, }d(H_{i},H_{j})\geq2\right\}  .
\]
We set $\Pi_{1}:A\times B\rightarrow A$, respectively $\Pi_{2}:A\times
B\rightarrow B$, for the canonical projections, and define $\widetilde{D}%
=\Pi_{1}D_{2}\cup\Pi_{2}D_{2}$. Any coloring $C$ satisfying $\mathfrak{F}%
\left(  G^{C}\right)  =H$ is determined by a set conditions of the following
form. There exists a partition $\mathcal{P}\left(  \widetilde{D}\right)
=\left\{  \widetilde{D}_{1},\ldots,\widetilde{D}_{k}\right\}  $, with
$\left\vert \widetilde{D}_{i}\right\vert \geq1$ for $i=1,\ldots,k$, such that%
\begin{equation}
\left\{  C(H_{i})\neq C(H_{j})\text{ for }d\left(  H_{i},H_{j}\right)
=1;\right.  \label{Coloring_conditions_A}%
\end{equation}%
\begin{equation}
\left\{
\begin{array}
[c]{l}%
C(H_{i})=C(H_{j})=b_{l}\in\mathbb{F}_{p}\text{, for any }\left\{  H_{i}%
,H_{j}\right\}  \in\widetilde{D}_{i},\\
\text{with }b_{l}\neq b_{m}\text{ if }l\neq m,\text{ for }l,m\in\left\{
1,\ldots,k\right\}  .
\end{array}
\right.  \label{Coloring_conditions_B}%
\end{equation}
The set of conditions (\ref{Coloring_conditions_A}%
)-(\ref{Coloring_conditions_B}) defines a relative closed subset of the affine
space $\mathbb{F}_{p}^{M}$, for a suitable $M$, \ and the solution set of
these conditions \ corresponds to the colorings defined by conditions
(\ref{Coloring_conditions_A})-(\ref{Coloring_conditions_B}).
\end{remark}

\begin{example}
In this example, we compute the chromatic function $\mathcal{C}(p;H)$, where
$H$ is in $Subgraphs(G,\left\vert G\right\vert )$, with $G$ and $H$ as follows:

\begin{figure}[pth]
\hskip 4cm \epsfxsize=6cm \epsfbox{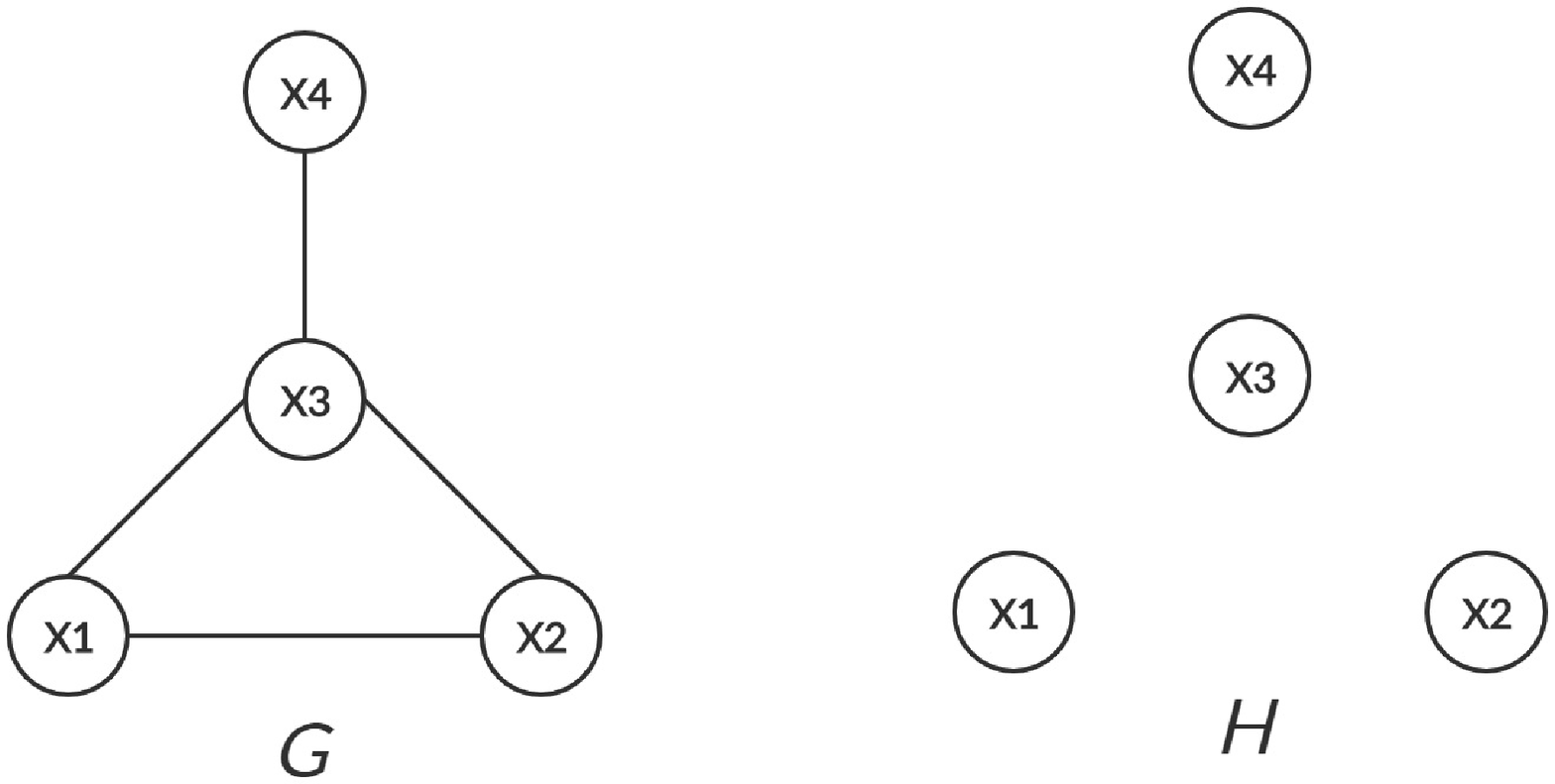}\end{figure}

In this case $H=H_{1}\#\cdots\#H_{4}$, where $H_{i}=\left\{  x_{i}\right\}  $
is the vertex $x_{i}$, for $i=1,2,3,4$. Set $C(H_{i})=a_{i}$, for $i=1,2,3,4$.
There are three different types of conditions (colorings) \ coming from
$\mathfrak{F}(G^{C})=H$:%
\begin{equation}
\left\{
\begin{array}
[c]{l}%
a_{1}\neq a_{2}\text{, }a_{1}\neq a_{3}\text{, }a_{2}\neq a_{3}\text{, }%
a_{3}\neq a_{4};\\
\\
a_{1}\neq a_{4}\text{, }a_{2}\neq a_{4}\text{; }%
\end{array}
\right.  \label{Coloring_1}%
\end{equation}%
\begin{equation}
\left\{
\begin{array}
[c]{l}%
a_{1}\neq a_{2}\text{, }a_{1}\neq a_{3}\text{, }a_{2}\neq a_{3}\text{, }%
a_{3}\neq a_{4};\\
\\
a_{1}=a_{4}\text{. }%
\end{array}
\right.  \label{Coloring_2}%
\end{equation}%
\begin{equation}
\left\{
\begin{array}
[c]{l}%
a_{1}\neq a_{2}\text{, }a_{1}\neq a_{3}\text{, }a_{2}\neq a_{3}\text{, }%
a_{3}\neq a_{4};\\
\\
a_{2}=a_{4}\text{. }%
\end{array}
\right.  \label{Coloring_3}%
\end{equation}
Consequently
\[
C(p,H)=p(p-1)(p-2)(p-3)+2p(p-1)(p-2),
\]
for any prime number $p$.

We now explain the connection between chromatic functions and the computation
of certain $p$-adic integrals. Set
\[
F_{G}\left(  \boldsymbol{x},\boldsymbol{s}\right)  =\left\vert x_{1}%
-x_{2}\right\vert _{p}^{s_{12}}\left\vert x_{1}-x_{3}\right\vert _{p}^{s_{13}%
}\left\vert x_{2}-x_{3}\right\vert _{p}^{s_{23}}\left\vert x_{3}%
-x_{4}\right\vert _{p}^{s_{24}},
\]
and
\[
I(\boldsymbol{s},\boldsymbol{a})=\int\limits_{\boldsymbol{a}+p\mathbb{Z}%
_{p}^{4}}F_{G}(\boldsymbol{x},\boldsymbol{s})\prod\limits_{i=1}^{4}dx_{i},
\]
where $\boldsymbol{a}=\left(  a_{1},a_{2},a_{3},a_{4}\right)  \in
\mathbb{F}_{p}^{4}$. Assume that $\boldsymbol{a}$ is a coloring of one the
types (\ref{Coloring_1})-(\ref{Coloring_3}), i.e. $\boldsymbol{a}$ is a
solution of exactly one of the conditions systems (\ref{Coloring_1}%
)-(\ref{Coloring_3}), then \ by using that
\begin{gather*}
\left\vert a_{1}-a_{2}-p\left(  x_{1}-x_{2}\right)  \right\vert _{p}^{s_{12}%
}\left\vert a_{1}-a_{3}-p\left(  x_{1}-x_{3}\right)  \right\vert _{p}^{s_{13}%
}\left\vert a_{2}-a_{3}-p\left(  x_{2}-x_{3}\right)  \right\vert _{p}^{s_{23}%
}\times\\
\left\vert a_{3}-a_{4}-p\left(  x_{3}-x_{4}\right)  \right\vert _{p}^{s_{24}%
}=1\text{, for any }x_{1},x_{2},x_{3},x_{4},
\end{gather*}
we have $I(\boldsymbol{s},\boldsymbol{a})=p^{-4}$. Now notice that%
\[
\left\vert \left\{  \boldsymbol{a}\in\mathbb{F}_{p}^{4};I(\boldsymbol{s}%
,\boldsymbol{a})=p^{-4}\right\}  \right\vert =C(p,H)\text{ for any prime
number }p\text{.}%
\]

\end{example}

\begin{remark}
We review the classical definitions of vertex colorings and chromatic
polynomial. Let $G$ be a graph and let $k$ be a positive integer. A proper
$k$-coloring of the vertices of $G$ is a function $f:V(G)\rightarrow\left\{
0,\ldots,k-1\right\}  $ \ such that $f^{-1}\left(  j\right)  $ is an
independent set, i.e. for any $u$,$v$ $\in f^{-1}\left(  j\right)  $ there is
no edge in $E(G)$ joining them. Let $\mathcal{P}(k;G)$ denotes the number of
vertex $k$-colorings of $G$. There exists a polynomial $\mathcal{P}(x;G)$ (the
chromatic polynomial of $G$), with integer coefficients, satisfying
$\mathcal{P}(x;G)\mid_{x=k}=\mathcal{P}(k;G)$ for any positive integer $k$,
see e.g. \cite[Proposition 9.2]{Biggs}. The chromatic number $\chi\left(
G\right)  $\ of $G$ is the positive integer defined as $\chi\left(  G\right)
=\min\left\{  k\in\mathbb{N\smallsetminus}\left\{  0\right\}  ;\mathcal{P}%
(k;G)>0\right\}  $.
\end{remark}

\begin{definition}
\label{Definition_8}Let $H$ be a subgraph in $Subgraphs(G,\left\vert
G\right\vert )$, such that $H=H_{1}\#\cdots\#H_{r}$, where the $\ H_{i}$s are
the different connected components of $H$. We attach to $H$ the graph
$G_{H}^{\ast}$ defined as follows:%
\[
V(G_{H}^{\ast})=\left\{  H_{1},\cdots,H_{r}\right\}  \text{, and \ }%
E(G_{H}^{\ast})=\left\{  \left\{  H_{i},H_{j}\right\}  ;d\left(  H_{i}%
,H_{j}\right)  =1\right\}  .
\]

\end{definition}

\begin{proposition}
\label{Prop2}For any graph $G$ and any $H$ in $Subgraphs(G,\left\vert
G\right\vert )$, $\mathcal{C}(p;H)=\mathcal{P}(x;G_{H}^{\ast})\mid_{x=p}$.
\end{proposition}

\begin{proof}
We assume that $H=H_{1}\#\cdots\#H_{r}$ as in Definition \ref{Definition_8}.
The result follows by establishing a bijection between the following two
sets:
\[
A\left(  G^{C},H\right)  :=\left\{  C\in Colors(G);\mathfrak{F}\left(
G^{C}\right)  =H\right\}  ,
\]%
\[
B\left(  G_{H}^{\ast}\right)  :=\left\{  p\text{-colorings of }G_{H}^{\ast
}\right\}  .
\]
Given a coloring $C\in A\left(  G^{C},H\right)  $, we define%
\[%
\begin{array}
[c]{llll}%
C^{\ast}: & V(G_{H}^{\ast}) & \rightarrow & \left\{  0,\ldots,p-1\right\} \\
& H_{i} & \rightarrow & C\left(  H_{i}\right)  .
\end{array}
\]
Now, if $C_{1}$, $C_{2}\in A\left(  G^{C},H\right)  $ and $C_{1}\neq C_{2}$,
then there exists $j\in\left\{  1,\ldots,r\right\}  $ such that $\left.
C_{1}\right\vert _{H_{j}}\neq\left.  C_{2}\right\vert _{H_{j}}$ which implies
that $C_{1}^{\ast}\neq C_{2}^{\ast}$.

Given a $p$-coloring $C^{\ast}$ of $G_{H}^{\ast}$, we define
\[%
\begin{array}
[c]{llll}%
C: & V(G) & \rightarrow & \left\{  0,\ldots,p-1\right\} \\
& v & \rightarrow & C^{\ast}(H_{i}),
\end{array}
\]
for any $v\in H_{i}$. Then $C\in A\left(  G^{C},H\right)  $. Indeed, by the
definition of $C$, $G^{C}=H_{1}\#\cdots\#H_{r}=H$, with $\left.  C\right\vert
_{H_{i}}=a_{i}\in\mathbb{F}_{p}$ for $i=1,\ldots,r$. Then $V(G^{C})=V(H)$.
Additionally, an edge $l\in E\left(  G^{C}\right)  $ is and edge of $G$, say
$i_{G}(l)=\left\{  u,v\right\}  $, satisfying $C(u)=C(v)$. Then $u$, $v\in
V(H_{i})$, and $l\in E(H_{i})$, i.e. $E(G^{C})\subset E(H)$. Conversely, given
$l\in E(H_{i})$, with $i_{H}(l)=\left\{  u,v\right\}  $, we have
$C(u)=C(v)=C^{\ast}(H_{i})$, and thus $l\in V(G^{C})$.
\end{proof}

\subsection{Rationality and recursive formulas}

\begin{theorem}
\label{Theorem1}Let $G$ be a connected graph. Then, for any prime number $p$,
$Z\left(  \boldsymbol{s};G\right)  $ satisfies:

\noindent(i)
\[
Z(\boldsymbol{s};G)=\frac{%
%TCIMACRO{\dsum \limits_{\substack{_{H\in Subgraphs_{\mathcal{F}}(G,\left\vert
%G\right\vert )}\\H\neq G}}}%
%BeginExpansion
{\displaystyle\sum\limits_{\substack{_{H\in Subgraphs_{\mathcal{F}%
}(G,\left\vert G\right\vert )}\\H\neq G}}}
%EndExpansion
p^{-\left\vert V\left(  G\right)  \right\vert -\sum_{l\in E\left(  H\right)
}s\left(  l\right)  }\mathcal{C}(p;H)Z\left(  \boldsymbol{s};H\right)
}{1-p^{1-\left\vert V\left(  G\right)  \right\vert -\sum_{l\in E\left(
G\right)  }s\left(  l\right)  }}.
\]

\noindent(ii) $Z\left(  \boldsymbol{s};G\right)  $ admits a meromorphic
continuation to $\mathbb{C}^{\left\vert E\left(  G\right)  \right\vert }$ as a
rational function of $\left\{  p^{-s\left(  l\right)  };l\in E\left(
G\right)  \right\}  $. More precisely,%
\begin{equation}
Z\left(  \boldsymbol{s};G\right)  =\frac{M\left(  \left\{  p^{-s\left(
l\right)  };l\in E\left(  G\right)  \right\}  \right)  }{\prod
\limits_{\substack{H\in Indgraphs(G)\\\left\vert V\left(  H\right)
\right\vert \geq2}}\left(  1-p^{1-\left\vert V\left(  H\right)  \right\vert
-\sum_{l\in E\left(  H\right)  }s\left(  l\right)  }\right)  },
\label{Eq_Zeta_Num}%
\end{equation}
where $M\left(  \left\{  p^{-s\left(  l\right)  };l\in E\left(  G\right)
\right\}  \right)  $ denotes a polynomial with rational coefficients in the
variables $\left\{  p^{-s\left(  l\right)  }\right\}  _{l\in E\left(
G\right)  }$.
\end{theorem}

\begin{proof}
(i) We attach to $\boldsymbol{a}=\left\{  a_{v}\right\}  _{v\in V(G)}%
\in\mathbb{F}_{p}^{\left\vert V(G)\right\vert }$ a color $C$ defined as
$C(v)=a_{v}$, for $v\in V(G)$. We set%
\[
I\left(  s;\boldsymbol{a}\right)  :=\int\limits_{\boldsymbol{a+}%
p\mathbb{Z}_{p}^{^{\left\vert V(G)\right\vert }}}F_{G}\left(  \boldsymbol{x}%
,\boldsymbol{s}\right)  \prod\limits_{v\in V(G)}dx_{v},
\]
then
\[
Z\left(  \boldsymbol{s};G\right)  =\sum\limits_{\boldsymbol{a}\in
\mathbb{F}_{p}^{\left\vert V(G)\right\vert }}I\left(  s;\boldsymbol{a}\right)
\text{.}%
\]
Now
\[
I\left(  s;\boldsymbol{a}\right)  =p^{-\left\vert V(G)\right\vert }%
%TCIMACRO{\dint \limits_{\mathbb{Z}_{p}^{^{\left\vert V(G)\right\vert }}}}%
%BeginExpansion
{\displaystyle\int\limits_{\mathbb{Z}_{p}^{^{\left\vert V(G)\right\vert }}}}
%EndExpansion
F_{G}\left(  \boldsymbol{a}+p\boldsymbol{x},\boldsymbol{s}\right)
\prod\limits_{v\in V(G)}dx_{v},
\]
where
\begin{align*}
F_{G}\left(  \boldsymbol{a}+p\boldsymbol{x},\boldsymbol{s}\right)   &
=\prod\limits_{\substack{l\in E\left(  G\right)  \\i_{G}(l)=\left\{
v,u\right\}  }}\left\vert a_{v}-a_{u}+px_{v}-px_{u}\right\vert _{p}^{s\left(
l\right)  }\\
&  =\prod\limits_{\substack{l\in E\left(  G\right)  \\i_{G}(l)=\left\{
v,u\right\}  }}\left\{
\begin{array}
[c]{lll}%
1 & \text{if} & C(v)\neq C(u)\\
&  & \\
p^{-s\left(  l\right)  }\left\vert x_{v}-x_{u}\right\vert _{p}^{s\left(
l\right)  } & \text{if} & C(v)=C(u).
\end{array}
\right.
\end{align*}
By attaching to $I\left(  s;\boldsymbol{a}\right)  $ the colored graph
$G^{C}=\left(  G^{C}\right)  _{\text{red}}\#\left(  G^{C}\right)
^{\text{iso}}$, and using $G_{\text{red}}^{C}=\left(  G^{C}\right)
_{\text{red}}$ by simplicity, we have
\[
F_{G}\left(  \boldsymbol{a}+p\boldsymbol{x},\boldsymbol{s}\right)
=p^{-\sum_{l\in E\left(  G_{\text{red}}^{C}\right)  }s\left(  l\right)  }%
\prod\limits_{\substack{l\in E\left(  G_{\text{red}}^{C}\right)
\\i_{G}(l)=\left\{  v,u\right\}  }}\left\vert x_{v}-x_{u}\right\vert
_{p}^{s\left(  l\right)  },
\]
and%
\[
I\left(  \boldsymbol{s};\boldsymbol{a}\right)  =p^{-\left\vert V\left(
G\right)  \right\vert -\sum_{l\in E\left(  G_{\text{red}}^{C}\right)
}s\left(  l\right)  }Z\left(  \left\{  s\left(  l\right)  \right\}  _{l\in
E\left(  G_{\text{red}}^{C}\right)  }\text{,}\left\{  x_{v}\right\}  _{v\in
V\left(  G_{\text{red}}^{C}\right)  }\right)  .
\]
Therefore
\[
Z\left(  \boldsymbol{s};G\right)  =%
%TCIMACRO{\dsum \limits_{_{G^{C}\text{, }C\in Colors(G)}}}%
%BeginExpansion
{\displaystyle\sum\limits_{_{G^{C}\text{, }C\in Colors(G)}}}
%EndExpansion
p^{-\left\vert V\left(  G\right)  \right\vert -\sum_{l\in E\left(
G_{\text{red}}^{C}\right)  }s\left(  l\right)  }Z\left(  \boldsymbol{s}%
;G_{\text{red}}^{C}\right)  .
\]
By fixing a graph $H$ in $Subgraphs_{\mathcal{F}}(G,\left\vert G\right\vert
)$, we have
\begin{gather}%
%TCIMACRO{\dsum \limits_{_{\mathcal{F}\left(  G^{C}\right)  =H}}}%
%BeginExpansion
{\displaystyle\sum\limits_{_{\mathcal{F}\left(  G^{C}\right)  =H}}}
%EndExpansion
p^{-\left\vert V\left(  G\right)  \right\vert -\sum_{l\in E\left(
G_{\text{red}}^{C}\right)  }s\left(  l\right)  }Z\left(  \boldsymbol{s}%
;G_{\text{red}}^{C}\right)  =\label{Zeta_Formula_4}\\
p^{-\left\vert V\left(  G\right)  \right\vert -\sum_{l\in E\left(  H\right)
}s(l)}\mathcal{C}(p;H)Z\left(  \boldsymbol{s};H\right)  ,\nonumber
\end{gather}
and consequently%
\begin{equation}
Z\left(  \boldsymbol{s};G\right)  =%
%TCIMACRO{\dsum \limits_{_{H\in Subgraphs_{\mathcal{F}}(G,\left\vert
%G\right\vert )}}}%
%BeginExpansion
{\displaystyle\sum\limits_{_{H\in Subgraphs_{\mathcal{F}}(G,\left\vert
G\right\vert )}}}
%EndExpansion
p^{-\left\vert V\left(  G\right)  \right\vert -\sum_{l\in E\left(  H\right)
}s\left(  l\right)  }\mathcal{C}(p;H)Z\left(  \boldsymbol{s};H\right)
\label{Zeta_Formula_2}%
\end{equation}
By taking $H=G$, $\mathcal{C}(p;H)=p$, in (\ref{Zeta_Formula_4}), we get
\[%
%TCIMACRO{\dsum \limits_{_{\mathcal{F}\left(  G^{C}\right)  =G}}}%
%BeginExpansion
{\displaystyle\sum\limits_{_{\mathcal{F}\left(  G^{C}\right)  =G}}}
%EndExpansion
p^{-\left\vert V\left(  G\right)  \right\vert -\sum_{l\in E\left(  G\right)
}s\left(  l\right)  }Z\left(  s;G^{C}\right)  =p^{1-\left\vert V\left(
G\right)  \right\vert -\sum_{l\in E\left(  G\right)  }s\left(  l\right)
}Z\left(  \boldsymbol{s};G\right)
\]
and thus from (\ref{Zeta_Formula_2}),\
\begin{equation}
Z(\boldsymbol{s};G)=\frac{%
%TCIMACRO{\dsum \limits_{\substack{_{H\in Subgraphs_{\mathcal{F}}(G,\left\vert
%G\right\vert )}\\H\neq G}}}%
%BeginExpansion
{\displaystyle\sum\limits_{\substack{_{H\in Subgraphs_{\mathcal{F}%
}(G,\left\vert G\right\vert )}\\H\neq G}}}
%EndExpansion
p^{-\left\vert V\left(  G\right)  \right\vert -\sum_{l\in E\left(  H\right)
}s\left(  l\right)  }\mathcal{C}(p;H)Z\left(  \boldsymbol{s};H\right)
}{1-p^{1-\left\vert V\left(  G\right)  \right\vert -\sum_{l\in E\left(
G\right)  }s\left(  l\right)  }}. \label{Zeta_Formula_3}%
\end{equation}
Now, taking $H=H_{1}\#\cdots\#H_{r(H)}\#H^{\text{iso}}$, where the $H_{i}$s
are different graphs in $Indgraphs(H)$, we have
\begin{equation}
Z\left(  \boldsymbol{s};H\right)  =\prod\nolimits_{j=1}^{r(H)}Z(\boldsymbol{s}%
;H_{j}). \label{Zeta_Formula_3A}%
\end{equation}
By using recursively (\ref{Zeta_Formula_3})-(\ref{Zeta_Formula_3A}), and the
formula for $Z(\boldsymbol{s};K_{2})$, we obtain (\ref{Eq_Zeta_Num}). Notice
that at the beginning of any iteration of the formulas (\ref{Zeta_Formula_3}%
)-(\ref{Zeta_Formula_3A}), with $\left\vert H_{j}\right\vert \geq2$ for
$j=1,\ldots,r(H)$, we have%
\[
\prod\nolimits_{j=1}^{r(H)}Z(\boldsymbol{s};H_{j})=\frac{A(\boldsymbol{s}%
;H_{1},\ldots,H_{r(H)})}{\prod\nolimits_{j=1}^{r(H)}\left(  1-p^{1-\left\vert
V\left(  H_{j}\right)  \right\vert -\sum_{l\in E\left(  H_{j}\right)
}s\left(  l\right)  }\right)  },
\]
where all the factors in the denominator are different since $H_{j}\cap
H_{i}=\varnothing$ if $j\neq i$.
\end{proof}

\begin{corollary}
\label{Cor1}(i) Set $s\left(  l\right)  =\gamma\in\mathbb{C}$ for any $l\in
E(G)$, and define $\mathcal{Z}_{G,p}\left(  \gamma\right)  :=\left.  Z\left(
\boldsymbol{s};G\right)  \right\vert _{s\left(  l\right)  =\gamma}$. Then the
integral $\mathcal{Z}_{G,p}\left(  \gamma\right)  $ converges for
\[
\operatorname{Re}(\gamma)\geq\max_{_{\substack{H\in Indgraphs(G)\\\left\vert
V\left(  H\right)  \right\vert \geq2}}}\frac{1-\left\vert V(H)\right\vert
}{\left\vert E(H)\right\vert }=:\gamma_{0}.
\]
More generally, for $G$ and $p$ fixed, $\mathcal{Z}_{G,p}\left(
\gamma\right)  $ \ is an analytic function in $\gamma$ for $\operatorname{Re}%
(\gamma)\geq\gamma_{0}.$

\noindent(ii) Let $G=K_{N}$ be the complete graph with $N$ vertices. Then
$\mathcal{Z}_{G,p}\left(  \gamma\right)  $ \ is an analytic function in
$\gamma$ for $\operatorname{Re}(\gamma)\geq\frac{-2}{N}.$

\noindent(iii) Let $M\left(  \left\{  p^{-s\left(  l\right)  };l\in E\left(
G\right)  \right\}  \right)  $ be the \ polynomial defined in
(\ref{Eq_Zeta_Num}). Then the following functional equations hold true:
\[
M\left(  \left\{  p^{-s\left(  l\right)  };l\in E\left(  G\right)  \right\}
\right)  =M\left(  \left\{  p^{-s\left(  \sigma_{E}\left(  l\right)  \right)
};l\in E\left(  G\right)  \right\}  \right)
\]
for any $\sigma=\left(  \sigma_{V},\sigma_{E}\right)  \in$Aut$\left(
G\right)  $.
\end{corollary}

\begin{proof}
(i) It follows directly from Theorem \ref{Theorem1}-(ii), by using the
properties of the geometric series. (ii) It follows from the fact that any
induced subgraph $H$ of $K_{N}$ is complete, say $H=K_{l}$, $\left\vert
V(H)\right\vert =l$, $\left\vert E(H)\right\vert =\frac{l\left(  l-1\right)
}{2}$ for $l=2,\ldots,N$. Then
\[
\gamma_{0}=\max_{2\leq l\leq N}\frac{-2}{l}=\frac{-2}{N}.
\]
(iii) It follows from Theorem \ref{Theorem1}-(ii) and Corollary \ref{Cor1_A}
by using the fact that any isomorphism of $G$ induces a permutation on the set
$\left\{  H\in Indgraphs(G);\left\vert V\left(  H\right)  \right\vert
\geq2\right\}  $.
\end{proof}

\begin{corollary}
\label{Cor2}(i) Let $G_{I}$ be an Indgraph of $G$ generated by $I\subset
V(G)$. Then%
\[
Z(\boldsymbol{s};G_{I})=Z(\boldsymbol{s};G)\mid_{\substack{s(l)=0\\l\notin
E(G_{I})}}.
\]
\noindent(ii) If $\lim_{s(l)\rightarrow a_{l}}Z(\boldsymbol{s};G_{I})=\infty$,
then
\[
\lim_{_{\substack{s(l)\rightarrow a_{l}\\l\in E(G_{I})}}}\lim
_{_{\substack{s(l)\rightarrow0\\l\notin E(G_{I})}}}Z(\boldsymbol{s}%
;G)=\infty.
\]
\noindent(iii) Let $l_{0}\in E(G)$ and let $K_{2}$ be the corresponding
induced graph. Then%
\[
\lim_{_{\substack{s(l)\rightarrow0\\l\in E(G_{I})\smallsetminus\left\{
l_{0}\right\}  }}}\lim_{s(l_{0})\rightarrow\text{ }-1}Z(\boldsymbol{s}%
;G)=\infty.
\]

\end{corollary}

\begin{proof}
(i) It follows from Theorem \ref{Theorem1}-(i). (ii) It follows from (i).
(iii) It follows from (ii) by using the formula for $Z(s,K_{2})$.
\end{proof}

\section{\label{Section_4}Phase transitions at finite temperature I}

\subsection{Log-Coulomb gases on graphs}

Let $G$ be a graph \ as before. Consider a log-Coulomb gas consisting of
$\left\vert V(G)\right\vert $ charges, $e_{v}\in\mathbb{R}$ for $v\in V(G)$,
which are located at $x_{v}\in\mathbb{Z}_{p}$ for $v\in V(G)$. We set as in
the introduction $\boldsymbol{x}=\left\{  x_{v}\right\}  _{v\in V(G)}%
\in\mathbb{Z}_{p}^{\left\vert V(G)\right\vert }$, $\mathbf{e}_{G}=\left\{
e_{v}\right\}  _{v\in V(G)}\in\mathbb{R}^{\left\vert V(G)\right\vert }$. The
Hamiltonian of the gas is%
\[
H_{p}(\boldsymbol{x};\mathbf{e},\beta,G)=-%
%TCIMACRO{\tsum \limits_{\substack{u,v\in V\left(  G\right)  \\u\sim v}}}%
%BeginExpansion
{\textstyle\sum\limits_{\substack{u,v\in V\left(  G\right)  \\u\sim v}}}
%EndExpansion
\ln\left\vert x_{u}-x_{v}\right\vert _{p}^{e_{u}e_{v}}+\frac{1}{\beta
}P(\boldsymbol{x}),
\]
where the confining potential is given by
\[
P(\boldsymbol{x})=\left\{
\begin{array}
[c]{cc}%
0 & \text{if }\left\{  x_{v}\right\}  _{v\in V(G)}\in\mathbb{Z}_{p}%
^{\left\vert V(G)\right\vert }\\
& \\
+\infty & \text{otherwise.}%
\end{array}
\right.
\]
The interaction between the two charged particles located at $x_{u}$ and
$x_{v}$ is only possible when $u\sim v$. This condition can be naturally
reformulated saying that the potential $V$ creates a potential well, supported
in $\mathbb{Z}_{p}^{\left\vert V(G)\right\vert }$, whose geometry corresponds
to the graph $G$.

The partition function of this gas is given by%
\[
\mathcal{Z}_{G,p,\mathbf{e}}\left(  \beta\right)  =%
%TCIMACRO{\dint \limits_{\mathbb{Z}_{p}^{\left\vert V(G)\right\vert }}}%
%BeginExpansion
{\displaystyle\int\limits_{\mathbb{Z}_{p}^{\left\vert V(G)\right\vert }}}
%EndExpansion
\text{ }%
%TCIMACRO{\tprod \limits_{_{\substack{u,v\in V\left(  G\right)  \\u\sim v}}}}%
%BeginExpansion
{\textstyle\prod\limits_{_{\substack{u,v\in V\left(  G\right)  \\u\sim v}}}}
%EndExpansion
\text{ \ }\left\vert x_{u}-x_{v}\right\vert _{p}^{e_{u}e_{v}\beta}%
%TCIMACRO{\tprod \limits_{v\in V\left(  G\right)  }}%
%BeginExpansion
{\textstyle\prod\limits_{v\in V\left(  G\right)  }}
%EndExpansion
dx_{v}=\left.  Z_{\varphi}(\boldsymbol{s};G)\right\vert _{s(u,v)=e_{u}%
e_{v}\beta},
\]
where $\varphi$ is the characteristic function of $\mathbb{Z}_{p}^{\left\vert
V(G)\right\vert }$. The statistical mechanics of the gas is described by the
corresponding Gibbs \ measure:%
\begin{align*}
d\mathbb{P}_{G,\beta,p,\mathbf{e}}(\boldsymbol{x})  &  =\frac{e^{-\beta
H_{p}(\boldsymbol{x};\mathbf{e},\beta,G)}}{\mathcal{Z}_{G,p,\mathbf{e}}\left(
\beta\right)  }%
%TCIMACRO{\tprod \limits_{v\in V\left(  G\right)  }}%
%BeginExpansion
{\textstyle\prod\limits_{v\in V\left(  G\right)  }}
%EndExpansion
dx_{v}\\
&  =\frac{%
%TCIMACRO{\tprod \limits_{_{\substack{u,v\in V\left(  G\right)  \\u\sim v}}}}%
%BeginExpansion
{\textstyle\prod\limits_{_{\substack{u,v\in V\left(  G\right)  \\u\sim v}}}}
%EndExpansion
\text{ \ }\left\vert x_{u}-x_{v}\right\vert _{p}^{e_{i}e_{j}\beta}%
}{\mathcal{Z}_{G,p,\mathbf{e}}\left(  \beta\right)  }1_{\mathbb{Z}%
_{p}^{\left\vert V(G)\right\vert }}(\left\{  x_{v}\right\}  _{v\in V(G)})%
%TCIMACRO{\tprod \limits_{v\in V\left(  G\right)  }}%
%BeginExpansion
{\textstyle\prod\limits_{v\in V\left(  G\right)  }}
%EndExpansion
dx_{v}.
\end{align*}
The probability measure $\mathbb{P}_{G,\beta,p,\mathbf{e}}(\boldsymbol{x})$
gives the probability of finding the particles at $\boldsymbol{x}$ at
temperature $\frac{1}{k_{B}\beta}$ given the charge distribution $\mathbf{e}$.

\subsection{Phase transitions I}

For $G,p,\mathbf{e}$ fixed, the partition function $\mathcal{Z}%
_{G,p,\mathbf{e}}\left(  \beta\right)  $ is a rational function in $p^{-\beta
}$ due to Theorem \ref{Theorem1}. The problem of determining the convergence
region for $\mathcal{Z}_{G,p,\mathbf{e}}\left(  \beta\right)  $ in terms of
the poles\ of the meromorphic continuation of $\mathcal{Z}_{G,p,\mathbf{e}%
}\left(  \beta\right)  $ is highly non-trivial. The integral $\mathcal{Z}%
_{G,p,\mathbf{e}}\left(  \beta\right)  $ converges when the following
conditions hold true:
\begin{equation}
1-\left\vert V(H)\right\vert -\sum\limits_{\substack{u,v\in V(H)\\u\sim
v}}e_{u}e_{v}\beta<0\text{ for }H\in Indgraphs(G)\text{, }\left\vert V\left(
H\right)  \right\vert \geq2\text{,} \label{Poles_1}%
\end{equation}
see Theorem \ref{Theorem1}-(ii). For $H\in Indgraphs(G)$, $\left\vert V\left(
H\right)  \right\vert \geq2$, we define%
\[
Char_{+}(H)=\sum\limits_{\substack{u,v\in V(H)\\u\sim v;\text{ }e_{u}e_{v}%
>0}}e_{u}e_{v};\text{ \ }Char_{-}(H)=\sum\limits_{\substack{u,v\in V(H)\\u\sim
v;\text{ }e_{u}e_{v}<0}}e_{u}e_{v}.
\]%
\[
Indgraphs_{-}(G):=\left\{  Indgraphs(G);Char_{+}(H)+Char_{-}(H)<0\right\}  ;
\]
and%
\[
Indgraphs_{+}(G):=\left\{  Indgraphs(G);Char_{+}(H)+Char_{-}(H)>0\right\}  .
\]
With this notation, we rewrite (\ref{Poles_1}) as%
\begin{multline*}
1-\left\vert V(H)\right\vert -\left\{  Char_{-}(H)+Char_{+}(H)\right\}
\beta<0\text{ for }H\in Indgraphs(G)\text{, }\left\vert V\left(  H\right)
\right\vert \geq2\text{,}\\
\text{ and }Char_{+}(H)+Char_{-}(H)\neq0.
\end{multline*}
Then the integral $\mathcal{Z}_{G,p,\mathbf{e}}\left(  \beta\right)  $
converges if%
\[
\left\{
\begin{array}
[c]{lll}%
\beta<\beta_{+}(H):=\frac{\left\vert V(H)\right\vert -1}{\left\vert
Char_{-}(H)+Char_{+}(H)\right\vert } & \text{for } &
\begin{array}
[c]{l}%
H\in Indgraphs_{-}(G)\text{, }\left\vert V\left(  H\right)  \right\vert
\geq2\text{,}\\
\text{ and }Char_{+}(H)+Char_{-}(H)\neq0;
\end{array}
\\
&  & \\
\beta>\beta_{-}(H):=\frac{1-\left\vert V(H)\right\vert }{Char_{+}%
(H)+Char_{-}(H)} & \text{for } &
\begin{array}
[c]{l}%
H\in Indgraphs_{+}(G)\text{, }\left\vert V\left(  H\right)  \right\vert
\geq2\text{,}\\
\text{and }Char_{+}(H)+Char_{-}(H)\neq0.
\end{array}
\end{array}
\right.
\]
If $Indgraphs_{-}(G)\neq\emptyset$ and $Indgraphs_{+}(G)\neq\emptyset$, \ we
set
\[
\beta_{UV}:=\min_{H\in Indgraphs_{-}(G)}\beta_{+}(H)\text{ \ and }\beta
_{IR}:=\max_{H\in Indgraphs_{+}(G)}\beta_{-}(H).
\]
If $Indgraphs_{-}(G)\neq\emptyset$ and $Indgraphs_{+}(G)=\emptyset,$ \ we set
\[
\beta_{UV}:=\min_{H\in Indgraphs_{-}(G)}\beta_{+}(H)\text{ \ and }\beta
_{IR}:=-\infty.
\]
If $Indgraphs_{-}(G)=\emptyset$ and $Indgraphs_{+}(G)\neq\emptyset,$ \ we set
\[
\text{ }\beta_{UV}:=+\infty\text{ and }\beta_{IR}:=\max_{H\in Indgraphs_{+}%
(G)}\beta_{-}(H).\text{ \ }%
\]
In this way we obtain the following result:

\begin{proposition}
\label{Prop3}With the above notation, the integral $\mathcal{Z}%
_{G,p,\mathbf{e}}\left(  \beta\right)  $ converges for \ $\beta_{IR}%
<\beta<\beta_{UV}$.
\end{proposition}

In order to decide wether or not $\mathcal{Z}_{G,p,\mathbf{e}}\left(
\beta\right)  $ converges for $\beta=\beta_{UV}$, we require an additional
condition. If the meromorphic continuation of $\mathcal{Z}_{G,p,\mathbf{e}%
}\left(  \beta\right)  $ has a pole at $\beta=\beta_{UV}$, then the integral
$\mathcal{Z}_{G,p,\mathbf{e}}\left(  \beta\right)  $ does not converge for
$\beta\geq\beta_{UV}$. \ 

\begin{remark}
Notice that Corollary \ref{Cor2} is not useful to determine phase transitions
points of $\mathcal{Z}_{G,p,\mathbf{e}}\left(  \beta\right)  $.
\end{remark}

If there exists $H\in Indgraphs(G)$ such that $\beta_{+}(H)>0$ and
$\mathcal{Z}_{G,p,\mathbf{e}}\left(  \beta\right)  $ has a pole at
$\beta=\beta_{+}(H)$, then by Proposition \ref{Prop3}, $\mathcal{Z}%
_{G,p,\mathbf{e}}\left(  \beta\right)  $ has a pole at the temperature
$\beta_{UV}$. Notice that $\beta_{UV}$ is not necessarily equal to $\beta
_{+}(H)$, since $\mathcal{Z}_{G,p,\mathbf{e}}\left(  \beta\right)  $ may have
other positive poles. In conclusion we have the following criteria:

\begin{proposition}
\label{Theorem4}With the above notation, and fixing $G,p,\mathbf{e}$. If there
exists $H\in Indgraphs(G)$ such that $\beta_{+}(H)>0$ and $\mathcal{Z}%
_{G,p,\mathbf{e}}\left(  \beta\right)  $ has a pole at $\beta=\beta_{+}(H)$,
then $\mathcal{Z}_{G,p,\mathbf{e}}\left(  \beta\right)  $ has a phase
transition at the temperature $\frac{1}{k_{B}\beta_{UV}}$.
\end{proposition}

\section{\label{Section_5A}The thermodynamic limit in star graphs}

The study of the thermodynamic limit in general graphs is a difficult matter
since it requires explicit formulas for the partition functions. In this
section we study the thermodynamic limit in star graphs, see Example
\ref{Example_Star}. We consider a neutral gas\ of $M=\left\vert V(S_{M}%
)\right\vert $ particles contained in a ball $B_{k}=p^{-k}\mathbb{Z}_{p}$, for
$k\in\mathbb{N}$. We assume that $M$ is even, and label the vertices of
$S_{M}$ as $V(S_{M})=\left\{  1,2,\ldots,\frac{M}{2},\frac{M}{2}%
+1,\ldots,M\right\}  $, where the vertex $1$ is the center of the star. We
assume a charge distribution $\mathbf{e}=\left\{  e_{i}\right\}  _{1\leq i\leq
M}$\ of the form%
\[
e_{i}=+1\text{ for }i=1,\ldots,\frac{M}{2}\text{ and \ }e_{i}=-1\text{ for
}i=\frac{M}{2}+1,\ldots,M.
\]
We label the edges as $V(E_{M})=\left\{  \left\{  1,2\right\}  ,\ldots
,\left\{  1,\frac{M}{2}\right\}  ,\left\{  1,\frac{M}{2}+1\right\}
,\ldots,\left\{  1,M\right\}  \right\}  $, and attach the complex variable
$s_{i}=s\left(  \left\{  1,i\right\}  \right)  $ of the edge $\left\{
1,i\right\}  $. Now we take $s_{i}=e_{1}e_{i}\beta$ \ for $i=1,\ldots,M$. \ We
denote the partition function attached to $\left(  M,\beta,p,\mathbb{Z}%
_{p},\mathbf{e}\right)  $ as $\mathcal{Z}_{M,0\text{ }}\left(  \beta\right)
$, and by $\mathcal{Z}_{M,k\text{ }}\left(  \beta\right)  $ the partition
function attached \ to $\left(  M,\beta,p,p^{-k}\mathbb{Z}_{p},\mathbf{e}%
\right)  $. Then by using Example \ \ref{Example_Star}, we get that%
\begin{gather}
\mathcal{Z}_{M,0\text{ }}\left(  \beta\right)  =Z(\boldsymbol{s};S_{N}%
)\mid_{s_{i}=e_{1}e_{i}\beta}=\int\limits_{\mathbb{Z}_{p}^{M}}\frac
{\prod\limits_{i=2}^{\frac{M}{2}}\left\vert x_{1}-x_{i}\right\vert _{p}%
^{\beta}}{\prod\limits_{i=1+\frac{M}{2}}^{M}\left\vert x_{1}-x_{i}\right\vert
_{p}^{\beta}}\prod\limits_{i=1}^{M}dx_{i}\label{Partition_function_0}\\
=\frac{\left(  1-p^{-1}\right)  ^{M-1}}{\left(  1-p^{-1-\beta}\right)
^{\frac{M}{2}-1}\left(  1-p^{-1+\beta}\right)  ^{\frac{M}{2}}}.\nonumber
\end{gather}
Notice that the integral in (\ref{Partition_function_0}) converges for
$\beta\in\left(  -1,1\right)  $. By Proposition \ref{Theorem4} $\ $there is a
phase transition at $\beta=1$. In order to determine $\mathcal{Z}_{M,k}\left(
\beta\right)  $ \ we use the Boltzmann factor $%
%TCIMACRO{\tprod \nolimits_{i=2}^{\frac{M}{2}}}%
%BeginExpansion
{\textstyle\prod\nolimits_{i=2}^{\frac{M}{2}}}
%EndExpansion
\left\vert x_{1}-x_{i}\right\vert _{p}^{\beta}%
%TCIMACRO{\tprod \nolimits_{i=1+\frac{M}{2}}^{M}}%
%BeginExpansion
{\textstyle\prod\nolimits_{i=1+\frac{M}{2}}^{M}}
%EndExpansion
\left\vert x_{1}-x_{i}\right\vert _{p}^{-\beta}$, then
\[
\mathcal{Z}_{M,k}\left(  \beta\right)  =\int\limits_{p^{-k}\mathbb{Z}_{p}^{M}%
}\frac{\prod\limits_{i=2}^{\frac{M}{2}}\left\vert x_{1}-x_{i}\right\vert
_{p}^{\beta}}{\prod\limits_{i=1+\frac{M}{2}}^{M}\left\vert x_{1}%
-x_{i}\right\vert _{p}^{\beta}}\prod\limits_{i=1}^{M}dx_{i}=p^{k\left(
-\beta+M\right)  }\mathcal{Z}_{M,0}\left(  \beta\right)  \text{, for }%
k\in\mathbb{N}\text{.}%
\]
The calculation of a thermodynamic limit requires to consider $M=\left\vert
V(S_{M})\right\vert $ $\rightarrow\infty$, $vol\left(  B_{k}\right)
=p^{k}\rightarrow\infty$, with $\frac{M}{p^{k}}=\rho$ fixed.

\subsection{The dimensionless free energy per particle}

Following the canonical formalism of statistical mechanics, we define the
total dimensionless free energy $\beta\mathfrak{F}$ as%
\[
\beta\mathfrak{F}=-\ln\left(  \frac{1}{\left\{  \left(  \frac{M-1}{2}\right)
!\right\}  ^{2}}\mathcal{Z}_{M,k}\left(  \beta\right)  \right)  .
\]
Notice that we use $\left\{  \left(  \frac{M-1}{2}\right)  !\right\}  ^{2}$
instead of $M!$. This is the cardinality of the elements of Aut($S_{N}$)
preserving the charge distribution on $S_{N}$. The dimensionless free energy
per particle $\beta\mathfrak{f}$ is defined as%
\[
\beta\mathfrak{f=}\lim_{\substack{M,vol(B_{k})\rightarrow\infty\\\frac
{M}{vol(B_{k})}=\rho}}\frac{1}{M}\beta\mathfrak{F.}%
\]
Then by using the Stirling formula we have $\frac{2\ln\left(  \frac{M-1}%
{2}\right)  !}{M}\sim-1+\ln\left(  \frac{M}{2}\right)  $, and%
\begin{align*}
\beta\mathfrak{f}  &  =\lim_{\substack{M,vol(B_{k}^{M})\rightarrow
\infty\\M=p^{k}\rho}}\frac{-1}{M}\ln\left(  \frac{1}{\left(  M-1\right)
!}p^{k\left(  -\beta+M\right)  }\frac{\left(  1-p^{-1}\right)  ^{M-1}}{\left(
1-p^{-1-\beta}\right)  ^{\frac{M}{2}-1}\left(  1-p^{-1+\beta}\right)
^{\frac{M}{2}}}\right) \\
&  =\lim_{\substack{M,vol(B_{k})\rightarrow\infty\\M=p^{k}\rho}}\left\{  -\ln
p^{k}-\ln\left(  1-p^{-1}\right)  +\frac{1}{2}\ln\left(  1-p^{-1-\beta
}\right)  \right. \\
&  \left.  +\frac{1}{2}\ln\left(  1-p^{-1+\beta}\right)  +\ln\left(  \frac
{M}{2}\right)  -1\right\} \\
&  =\ln\rho-\ln\left(  1-p^{-1}\right)  +\frac{1}{2}\ln\left(  1-p^{-1-\beta
}\right)  +\frac{1}{2}\ln\left(  1-p^{-1+\beta}\right)  -1-\ln2.
\end{align*}
Now, in the high temperature limit $\beta\rightarrow0$ and the $\beta
\mathfrak{f}$ tends to $\ln\rho-1$. There is phase transition \ at $\beta=1$,
since $\lim_{\beta\rightarrow1}\beta\mathfrak{f}=-\infty$.

The mean energy per particle%
\begin{align*}
\overline{E}  &  =\frac{\partial}{\partial\beta}\beta\mathfrak{f}%
=\lim_{\substack{M,vol(B_{k}^{M})\rightarrow\infty\\M=p^{kM}\rho}}\left(
\frac{1}{M}\left\langle H_{S_{M}}\right\rangle \right) \\
&  =\frac{\ln p}{2}\left(  \frac{p^{-1-\beta}}{1-p^{-1-\beta}}-\frac
{p^{-1+\beta}}{1-p^{-1+\beta}}\right)  .
\end{align*}

\subsection{The grand-canonical potential}

The grand-canonical distribution is the generating function for $\mathcal{Z}%
_{M,k}\left(  \beta\right)  $, for $k\in\mathbb{N}$, i.e.%
\begin{gather*}
\mathcal{Z}_{\beta,k}(X):=\sum\limits_{L=0}^{\infty}\mathcal{Z}_{2L,\beta
,k}X^{2L}:=1+\sum\limits_{L=1}^{\infty}\mathcal{Z}_{2L,\beta,k}X^{2L}\\
=1+\frac{\left(  1-p^{-1-\beta}\right)  }{\left(  1-p^{-1}\right)  }%
p^{-k\beta}\sum\limits_{L=1}^{\infty}\left\{  \frac{p^{k}\left(
1-p^{-1}\right)  X}{\sqrt{\left(  1-p^{-1-\beta}\right)  \left(
1-p^{-1+\beta}\right)  }}\right\}  ^{2L},
\end{gather*}
for $0\leq\beta<1$. Notice that $\mathcal{Z}_{\beta,k}(X)$ has a \ pole at
$\beta=1$, i.e. there is phase transition at $\frac{1}{k_{B}}$. Assuming that
\[
\frac{p^{k}\left(  1-p^{-1}\right)  \left\vert X\right\vert }{\sqrt{\left(
1-p^{-1-\beta}\right)  \left(  1-p^{-1+\beta}\right)  }}<1\text{, for }%
\beta\in\left[  0,1\right)  ,
\]
we have%
\begin{equation}
\mathcal{Z}_{\beta,k}(X)=1+\frac{p^{-k\left(  \beta-2\right)  }\left(
1-p^{-1}\right)  X^{2}}{\left(  1-p^{-1+\beta}\right)  }\left(  \frac
{1}{1-\left\{  \frac{p^{k}\left(  1-p^{-1}\right)  }{\sqrt{\left(
1-p^{-1-\beta}\right)  \left(  1-p^{-1+\beta}\right)  }}\right\}  ^{2}X^{2}%
}\right)  , \label{Eq_meromorphic}%
\end{equation}
for $\beta\in\left[  0,1\right)  $. We now use the meromorphic continuation
given \ (\ref{Eq_meromorphic}) to interpolate the values of $\mathcal{Z}%
_{\beta,k}(X)$ for any $X$, $k$, with $\beta\in\left[  0,1\right)  $. Now, we
compute the grand-canonical potential, by using that $\ln(1+z)\sim z$ as
$z\rightarrow0$,%
\[
\mathcal{X}_{\beta}(X)=\lim_{k\rightarrow\infty}\frac{1}{p^{k}}\ln
\mathcal{Z}_{\beta,k}(X)=0\text{ for }\beta\neq1\text{.}%
\]

\section{\label{Section_5}Local zeta functions for rational functions}

In the 70s Igusa developed a uniform theory for local zeta functions and
oscillatory integrals attached to polynomials with coefficients in a local
field of characteristic zero, \cite{Igusa}, \cite{Igusa-old}. In
\cite{Veys-Zuniga} this theory is extended to the case of rational functions.
We review some results of this article that are require here.

\subsection{Local fields of characteristic zero}

We take $\mathbb{K}$ to be a non-discrete$\ $locally compact field of
characteristic zero. Then $\mathbb{K}$ is $\mathbb{R}$, $\mathbb{C}$, or a
finite extension of $\mathbb{Q}_{p}$, the field of $p$-adic numbers. If
$\mathbb{K}$ is $\mathbb{R}$ or $\mathbb{C}$, we say that $\mathbb{K}$ is an
$\mathbb{R}$\textit{-field}, otherwise we say that $\mathbb{K}$ is a
$p$\textit{-field}.

For $a\in\mathbb{K}$, we define the \textit{modulus} $\left\vert a\right\vert
_{\mathbb{K}}$ of $a$ by%
\[
\left\vert a\right\vert _{K}=\left\{
\begin{array}
[c]{l}%
\text{the rate of change of the Haar measure in }(\mathbb{K},+)\text{ under
}x\rightarrow ax\text{ }\\
\text{for }a\neq0,\\
\\
0\ \text{ for }a=0\text{.}%
\end{array}
\right.
\]
It is well-known that, if $\mathbb{K}$ is an $\mathbb{R}$-field, then
$\left\vert a\right\vert _{\mathbb{R}}=\left\vert a\right\vert $ and
$\left\vert a\right\vert _{\mathbb{C}}=\left\vert a\right\vert ^{2}$, where
$\left\vert \cdot\right\vert $ denotes the usual absolute value in
$\mathbb{R}$ or $\mathbb{C}$, and, if $\mathbb{K}$ is a $p$-field, then
$\left\vert \cdot\right\vert _{\mathbb{K}}$ is the normalized absolute value
in $\mathbb{K}$.

\subsubsection{Structure of the $p$-fields}

A non-Archimedean local field $\mathbb{K}$ (or $p$-field) is a locally compact
topological field with respect to a non-discrete topology, which comes from a
norm $\left\vert \cdot\right\vert _{\mathbb{K}}$ satisfying
\[
\left\vert x+y\right\vert _{\mathbb{K}}\leq\max\left\{  \left\vert
x\right\vert _{\mathbb{K}},\left\vert y\right\vert _{\mathbb{K}}\right\}  ,
\]
for $x,y\in\mathbb{K}$. A such norm is called an \textit{ultranorm or
non-Archimedean}. Any non-Archimedean local field $\mathbb{K}$ of
characteristic zero is isomorphic (as a topological field) to a finite
extension of $\mathbb{Q}_{p}$, and it is called a $p$-adic field. The field
$\mathbb{Q}_{p}$ is the basic example of non-Archimedean local field of
characteristic zero.

The \textit{ring of integers} of $\mathbb{K}$ is defined as
\[
R_{\mathbb{K}}=\left\{  x\in\mathbb{K};\left\vert x\right\vert _{\mathbb{K}%
}\leq1\right\}  .
\]
Geometrically $R_{\mathbb{K}}$ is the unit ball of the normed space $\left(
\mathbb{K},\left\vert \cdot\right\vert _{\mathbb{K}}\right)  $. This ring is a
domain of principal ideals having a unique maximal ideal, which is given by
\[
P_{\mathbb{K}}=\left\{  x\in\mathbb{K};\left\vert x\right\vert _{\mathbb{K}%
}<1\right\}  .
\]
We fix a generator $\mathfrak{p}$ of $P_{\mathbb{K}}$ i.e. $P_{\mathbb{K}%
}=\mathfrak{p}R_{\mathbb{K}}$. A such generator is also called a \textit{local
uniformizing parameter of} $\mathbb{K}$, and it plays the same role as $p$ in
$\mathbb{Q}_{p}.$

The \textit{group of units} of $R_{\mathbb{K}}$ is defined as
\[
R_{\mathbb{K}}^{\times}=\left\{  x\in R_{\mathbb{K}};\left\vert x\right\vert
_{\mathbb{K}}=1\right\}  .
\]
The natural map $R_{\mathbb{K}}\rightarrow R_{\mathbb{K}}/P_{\mathbb{K}}%
\cong\mathbb{F}_{q}$ is called the \textit{reduction mod} $P_{\mathbb{K}}$.
The quotient $R_{\mathbb{K}}/P_{\mathbb{K}}\cong\mathbb{F}_{q}$, is a finite
field with $q=p^{f}$ elements, and it is called the \textit{residue field} of
$\mathbb{K}$. Every non-zero element $x$ of $\mathbb{K}$ can be written
uniquely as $x=\mathfrak{p}^{ord(x)}u$, $u\in R_{\mathbb{K}}^{\times}$. We set
$ord(0)=\infty$. The normalized valuation of $\mathbb{K}$ is the mapping
\[%
\begin{array}
[c]{ccc}%
\mathbb{K} & \rightarrow & \mathbb{Z}\cup\left\{  \infty\right\} \\
x & \rightarrow & ord(x).
\end{array}
\]
Then $\left\vert x\right\vert _{\mathbb{K}}=q^{-ord(x)}$ and $\left\vert
\mathfrak{p}\right\vert _{\mathbb{K}}=q^{-1}$.

We fix $\mathfrak{S}\subset R_{\mathbb{K}}$ a set of representatives of
$\mathbb{F}_{q}$ in $R_{\mathbb{K}}$, i.e. $\mathfrak{S}$ is a set which is
mapped bijectively onto $\mathbb{F}_{q}$ by the reduction $\operatorname{mod}$
$P_{\mathbb{K}}$. We assume that $0\in\mathfrak{S}$. Any non-zero element $x$
of $\mathbb{K}$ can be written as
\[
x=\mathfrak{p}^{ord(x)}\sum\limits_{i=0}^{\infty}x_{i}\mathfrak{p}^{i},
\]
where $x_{i}\in\mathfrak{S}$ and $x_{0}\neq0.$ This series converges in the
norm $\left\vert \cdot\right\vert _{\mathbb{K}}$.

\subsection{Local zeta functions for rational functions}

If $\mathbb{K}$ is a $p$-field, resp. an $\mathbb{R}$-field, we denote by
$\mathcal{D}(\mathbb{K}^{N})$ the $\mathbb{C}$-vector space consisting of all
$\mathbb{C}$-valued locally constant functions, resp. all smooth functions, on
$\mathbb{K}^{N}$, with compact support. An element of $\mathcal{D}%
(\mathbb{K}^{N})$ is called a \textit{test function}. To simplify terminology,
we will call a non-zero test function that takes only real and non-negative
values a \textit{positive}\ test function.

Let $f,g\in\mathbb{K}\left[  x_{1},\ldots,x_{N}\right]  \smallsetminus
\mathbb{K}$ be polynomial functions such that $f/g$ is not constant. Let
$\Phi:\mathbb{K}^{N}\rightarrow\mathbb{C}$ be a test function. Then the local
zeta function attached to $\left(  f/g,\Phi\right)  $ is defined as
\begin{equation}
Z_{\Phi}(s;f/g)=\int\limits_{\mathbb{K}^{N}\smallsetminus D_{\mathbb{K}}}%
\Phi\left(  x\right)  \left\vert \frac{f\left(  x\right)  }{g\left(  x\right)
}\right\vert _{\mathbb{K}}^{s}d^{N}x, \label{zeta}%
\end{equation}
where $\ s\in\mathbb{C}$, $D_{\mathbb{K}}=f^{-1}\left\{  0\right\}  \cup
g^{-1}\left\{  0\right\}  $ and $d^{N}x$ is the normalized Haar measure on
$\left(  \mathbb{K}^{N},+\right)  $. The convergence of the integral in
(\ref{zeta}) is not a straightforward matter; in particular the convergence
does not follow from the fact that $\Phi$ has compact support.

\subsubsection{Numerical data}

For any local field $\mathbb{K}$ of characteristic zero, there exits a finite
set of pair of integers depending on $\left(  f/g,\Phi\right)  $ of the form
\[
\left\{  \left(  n_{i},v_{i}\right)  ;i\in T_{+},n_{i}>0\right\}  \cup\left\{
\left(  n_{i},v_{i}\right)  ;i\in T_{-},n_{i}<0\right\}  ,
\]
where $T_{+}$ and $T_{-}$ are finite sets. We now define
\[
\alpha:=\alpha_{\Phi}=\left\{
\begin{array}
[c]{ll}%
\min_{i\in T_{-}}\left\{  \frac{v_{i}}{|n_{i}|}\right\}  & \text{if }T_{-}%
\neq\emptyset\\
& \\
+\infty & \text{if }T_{-}=\emptyset,
\end{array}
\right.
\]
and%
\[
\gamma:=\gamma_{\Phi}=\left\{
\begin{array}
[c]{ll}%
\max_{i\in T_{+}}\left\{  -\frac{v_{i}}{n_{i}}\right\}  & \text{if }T_{+}%
\neq\emptyset\\
& \\
-\infty & \text{if }T_{+}=\emptyset.
\end{array}
\right.
\]

\subsubsection{Meromorphic continuation: $p$-field\ case\textrm{ }}

Let $\mathbb{K}$ be a $p$-field. Then the following assertions hold: (1)
$Z_{\Phi}\left(  s;f/g\right)  $ converges for $\gamma<\operatorname{Re}%
(s)<\alpha$; (2) $Z_{\Phi}\left(  s;f/g\right)  $ has a meromorphic
continuation to $\mathbb{C}$ as a rational function of $q^{-s}$, and its poles
are of the form
\[
s=-\frac{v_{i}}{n_{i}}+\frac{2\pi\sqrt{-1}}{n_{i}\ln q}k,\ k\in\mathbb{Z},
\]
for $i\in T_{+}\cup T_{-}$. In addition, the order of any pole is at most $N$,
cf. \cite[Theorem 3.2]{Veys-Zuniga}.

\subsubsection{Meromorphic continuation: $\mathbb{R}$-field\ case}

Let $K$ be an $\mathbb{R}$-field. Then the following assertions hold: (1)
$Z_{\Phi}\left(  s;f/g\right)  $ converges for $\gamma<\operatorname{Re}%
(s)<\alpha$; (2) $Z_{\Phi}\left(  s,f/g\right)  $ has a meromorphic
continuation to $\mathbb{C}$, and its poles are of the form
\[
s=-\frac{v_{i}}{n_{i}}-\frac{k}{[\mathbb{K}:\mathbb{R}]n_{i}},\quad
k\in\mathbb{Z}_{\geq0},
\]
for $i\in T_{+}\cup T_{-}$, where $[\mathbb{K}:\mathbb{R}]=1$ if
$\mathbb{K}=\mathbb{R}$ and $[\mathbb{K}:\mathbb{R}]=2$ if $\mathbb{K}%
=\mathbb{C}$. In addition, the order of any pole is at most $N$, cf.
\cite[Theorem 3.5]{Veys-Zuniga}.

\subsection{Existence of poles, largest and smallest poles}

The theorems \cite[Theorems 3.5, 3.2]{Veys-Zuniga} above mentioned provide a
list of the possible poles for the local zeta $Z_{\Phi}\left(  s,f/g\right)  $
in terms of a list (the numerical data of resolution of singularities), which
is not unique neither intrinsic. Now, if $\gamma_{\Phi}\neq-\infty$, say it is
equal to $-\frac{v_{i}}{n_{i}}$ precisely for $i\in T_{\beta}\,(\subset
T_{+})$, then by choosing a suitable positive $\Phi$, $\beta_{\Phi}$ is a pole
of $Z_{\Phi}\left(  s;f/g\right)  $. And, if we assume that $\alpha_{\Phi}%
\neq+\infty$, say it is equal to $\frac{v_{i}}{|n_{i}|}$ precisely for $i\in
T_{\alpha}\,(\subset T_{-})$, then by choosing a suitable positive $\Phi$,
$\alpha_{\Phi}$ is a pole of $Z_{\Phi}\left(  s;f/g\right)  $, cf.
\cite[Theorem 3.9]{Veys-Zuniga}. This implies that $\alpha_{\Phi}$ and
$\gamma_{\Phi}$ do not depend on \ the numerical data used to compute them, if
we choose the test function $\Phi$ conveniently.

In \cite[Theorem 3.9]{Veys-Zuniga} some criteria for the existence of positive
and negative poles were developed. We review those criteria for the existence
of positive poles, since we use them later on. Let $U$ be an open subset of
$\mathbb{K}^{N}$. We assume that $\Phi$ is a test function with support
contained in $U$. If there exists a point $x_{0}\in U$ such that $f\left(
x_{0}\right)  \neq0$ and $g\left(  x_{0}\right)  =0$. Then, for any positive
test function $\Phi$ with support in a small enough neighborhood of $x_{0}$,
the zeta function $Z_{\Phi}\left(  s;f/g\right)  $ has a positive pole. In
particular, if $K=\mathbb{C}$ and $f$ and $g$ are polynomials, then $Z_{\Phi
}\left(  s;f/g\right)  $ always has a positive pole for an appropriate
positive test function $\Phi$, cf. \cite[Corollary 3.12]{Veys-Zuniga}.

\section{\label{Section_6}Phase transitions at finite temperature II}

In this section we consider a gas with \ $\left\vert V(G)\right\vert $
particles confined in a compact subset of $\mathbb{K}^{\left\vert
V(G)\right\vert }$, which is the support of a positive test function $\Phi$.
We assume that the charge distribution $\boldsymbol{e}=\left\{  e_{v}\right\}
_{v\in V(G)}$ satisfies the following hypothesis:%

\begin{equation}
\left\{
\begin{array}
[c]{l}%
e_{v}\in\left\{  +1,-1\right\}  \text{ for any }v\in V(G);\\
\\
\left\{  e_{v}e_{u};v,u\in V(G)\text{, }u\sim v\right\}  =\left\{
+1,-1\right\}  .
\end{array}
\right.  \tag{H1}%
\end{equation}
The Hamiltonian of this log-Coulomb gas has the form:%
\[
H_{\mathbb{K}}(\boldsymbol{x};\boldsymbol{e})=-%
%TCIMACRO{\tsum \limits_{\substack{u,v\in V\left(  G\right)  \\u\sim v}}}%
%BeginExpansion
{\textstyle\sum\limits_{\substack{u,v\in V\left(  G\right)  \\u\sim v}}}
%EndExpansion
\ln\left\vert x_{u}-x_{v}\right\vert _{\mathbb{K}}^{e_{u}e_{v}}-\frac{1}%
{\beta}\ln\Phi\left(  \boldsymbol{x}\right)  .
\]
The Boltzmann factor is $\exp\left(  -\beta H_{\mathbb{K}}(\boldsymbol{x}%
;\boldsymbol{e})\right)  $ and the partition function is%
\[
\mathcal{Z}_{G,\mathbb{K},\Phi,\mathbf{e}}\left(  \beta\right)  =\int
\limits_{\mathbb{K}^{\left\vert V(G)\right\vert }}\Phi\left(  \left\{
x_{v}\right\}  _{v\in V(G)}\right)  \frac{\prod\limits_{\substack{u,v\in
V(G)\\u\sim v;\text{ }e_{u}e_{v}=+1}}\left\vert x_{u}-x_{v}\right\vert
_{\mathbb{K}}^{\beta}}{\prod\limits_{\substack{u,v\in V(G)\\u\sim v;\text{
}e_{u}e_{v}=-1}}\left\vert x_{u}-x_{v}\right\vert _{\mathbb{K}}^{\beta}}%
\prod\limits_{u\in V(G)}dx_{u}.
\]
We now set
\[
f_{G,\mathbf{e}}(\boldsymbol{x}):=\prod\limits_{\substack{u,v\in V(G)\\u\sim
v;\text{ }e_{u}e_{v}=+1}}\left(  x_{u}-x_{v}\right)  \text{ and }%
g_{G,\mathbf{e}}(\boldsymbol{x}):=\prod\limits_{\substack{u,v\in V(G)\\u\sim
v;\text{ }e_{u}e_{v}=-1}}\left(  x_{u}-x_{v}\right)  .
\]
Assume that%
\begin{equation}
\text{there exists }\boldsymbol{x}_{0}\in\mathbb{K}^{\left\vert
V(G)\right\vert }\text{ such that }f_{G,\mathbf{e}}(\boldsymbol{x}_{0}%
)\neq0\text{ and }g_{G,\mathbf{e}}(\boldsymbol{x}_{0})=0. \tag{H2}%
\end{equation}
We pick a positive test function $\Phi$ supported in a small enough
neighborhood of $\boldsymbol{x}_{0}$. Then there exits $\beta_{UV}%
=\beta\left(  \Phi,G\right)  >0$ such that the integral $\mathcal{Z}%
_{G,\mathbb{K},\Phi,\mathbf{e}}\left(  \beta\right)  $ converges for $\beta
\in\left(  0,\beta_{UV}\right)  $, and the meromorphic continuation \ of
$\mathcal{Z}_{G,\mathbb{K},\Phi,\mathbf{e}}\left(  \beta\right)  $ has a pole
at $\beta=\beta_{UV}$.

\begin{theorem}
\label{Theorem5}Assume that $G,\mathbb{K},\Phi,\mathbf{e}$ are given, and the
Hypotheses H1, H2 hold. Then $\mathcal{Z}_{G,\mathbb{K},\Phi,\mathbf{e}%
}\left(  \beta\right)  $ has a phase transition at the temperature $\frac
{1}{k_{B}\beta_{UV}}$.
\end{theorem}

\textbf{Data Availability}. No new data were created or analyzed in this
study. Data sharing is not applicable to this article.

\begin{acknowledgement}
The first \ was partially supported by the Debnath Endowed Professorship,
UTRGV. The third author was partially supported by Conacyt Grant No. 286445.
On behalf of all authors, the corresponding author states that there is no
conflict of interest.
\end{acknowledgement}


\bigskip

\begin{thebibliography}{99}                                                                                               %


\bibitem {A-K-S}Albeverio S., Khrennikov A. Yu. and Shelkovich V. M.,
\textit{Theory of }$p$\textit{-adic distributions: linear and nonlinear
models} (Cambridge University Press, 2010).

\bibitem {Anashin}Anashin Vladimir, Khrennikov Andrei, \textit{Applied
algebraic dynamics}. De Gruyter Expositions in Mathematics, 49 (Walter de
Gruyter \& Co., 2009).

\bibitem {Av-4}V. A. Avetisov, A. Kh. Bikulov, V. A. Osipov, $\ $%
\textquotedblleft$p$-Adic description of characteristic relaxation in complex
systems,\textquotedblright\ \textit{J. Phys. A} \textbf{36}(15), 4239--4246 (2003).

\bibitem {Av-5}V. A. Avetisov, A. H. Bikulov, S. V. Kozyrev, V. A. Osipov,
\textquotedblleft$p$-Adic models of ultrametric diffusion constrained by
hierarchical energy landscapes,\textquotedblright\ \textit{J. Phys. A}
\textbf{35}(2), 177--189 (2002).

\bibitem {Balakrishnan et al}R. Balakrishnan and K. Ranganathan, \textit{A
textbook of graph theory. Second edition }(Springer, 2012).

\bibitem {Biggs}Biggs Norman, \textit{Algebraic graph theory. Second edition}
(Cambridge University Press, 1993).

\bibitem {Bocardo-Gaspar et al}Miriam Bocardo-Gaspar, Hugo
Garc\'{\i}a-Compe\'{a}n, Edgar Y. L\'{o}pez, and Wilson A.
Z\'{u}\~{n}iga-Galindo, \textquotedblleft Local Zeta Functions and
Koba--Nielsen String Amplitudes,\textquotedblright\ \textit{Symmetry}
\textbf{13}(6) 967 (2021). https://doi.org/10.3390/sym13060967.

\bibitem {Bourbaki}Bourbaki N. \textit{\'{E}l\'{e}ments de math\'{e}matique.
Fasc. XXXVI. Vari\'{e}t\'{e}s diff\'{e}rentielles et analytiques. Fascicule de
r\'{e}sultats (Paragraphes 8 \`{a} 15)} (Actualit\'{e}s Scientifiques et
Industrielles, No. 1347 Hermann, 1971).

\bibitem {Clair et al}B. Clair and S. Mokhtari-Sharghi, \textquotedblleft Zeta
functions of discrete groups acting on trees,\textquotedblright\ \textit{J.
Algebra} \textbf{237}, 591--620 (2001).

\bibitem {Connes-1}Alain \ Connes, \textquotedblleft Trace formula in
noncommutative geometry and the zeros of the Riemann zeta
function,\textquotedblright\ \textit{Selecta Math. (N.S.)} \textbf{5}(1),
29--106 (1999 ).

\bibitem {Connes-2}Alain Connes, \textquotedblleft Noncommutative geometry and
the Riemann zeta function,\textquotedblright\ \textit{Mathematics:}
\textit{Frontiers and Perspectives, 35--54}, Amer. Math. Soc., Providence, RI. 2000

\bibitem {Connes-Marcolli}Alain Connes and Matilde Marcolli,
\textit{Noncommutative geometry, quantum fields and motives} (Hindustan Book
Agency, 2008).

\bibitem {Denef}J. Denef, \textquotedblleft Report on Igusa's Local Zeta
Function,\textquotedblright\textit{ S\'{e}minaire Bourbaki}, Vol. 1990/91,
Exp. No.730-744 \textit{Ast\'{e}risque} 201-203, 359-386 (1991).

\bibitem {DL1}Jan \ Denef and Fran\c{c}ois Loeser, \textquotedblleft Motivic
Igusa zeta functions,\textquotedblright\ \textit{J. Algebraic Geom.}
\textbf{7}(3), 505--537 (1998).

\bibitem {Dorjevic et al}Goran S. Djordjevi\'{c}, Branko Dragovich,
Ljubi\v{s}a Ne\v{s}i\'{c}, \textquotedblleft Adelic path integrals for
quadratic Lagrangians,\textquotedblright\ \textit{Infin. Dimens. Anal. Quantum
Probab. Relat. Top.} \textbf{6}(2), 179--195 (2003).

\bibitem {D-K-K-V}B. Dragovich, A. Yu. Khrennikov, S. V. Kozyrev and I. V.
Volovich, \textquotedblleft On $p-$adic mathematical
physics,\textquotedblright\ $p-$\textit{adic Numbers} \textit{Ultrametric
Anal. Appl.} \textbf{1}(1), 1--17 (2009).

\bibitem {DysonFreeman}F. J. Dyson, \textquotedblleft An Ising ferromagnet
with discontinuous long-range order,\textquotedblright\ \textit{Comm. Math.
Phys.} \textbf{21}, 269--283 (1971).

\bibitem {Forrester et al}Peter J. Forrester and S. Ole Warnaar,
\textquotedblleft The importance of the Selberg integral,\textquotedblright%
\ \textit{Bull. Amer. Math. Soc. (N.S.)} \textbf{45}(4), 489--534 (2008).

\bibitem {Forrester}P. J. Forrester, \textit{Log-gases and random matrices
}(Princeton University Press, 2010).

\bibitem {G-S}I. M. Gel'fand and G.E. Shilov, \textit{Generalized Functions,
vol 1.} (Academic Press, 1997).

\bibitem {Grigorchuk et al}R. I. Grigorchuk and A. Zuk, \textquotedblleft The
Ihara zeta function of infinite graphs, the KNS spectral measure and
integrable maps,\textquotedblright\ \textit{in Random Walks and Geometry},
Walter de Gruyter, Berlin, pp. 141--180 (2004).

\bibitem {Gubser et al}S. S. Gubser, Ch. Jepsen, Z. Ji and B. Trundy,
\textquotedblleft Continuum limits of sparse coupling
patterns,\textquotedblright\ \textit{Phys. Rev. D} \textbf{98}(4), 045009, (2018).

\bibitem {Guido et al}Daniele Guido, Tommaso Isola and Michel L. Lapidus,
\textquotedblleft Ihara zeta functions for periodic simple
graphs,\textquotedblright\ \textit{C*-algebras and elliptic theory II}, Trends
Math., Birkh\"{a}user, Basel. pp 103--121 (2008).

\bibitem {Knauf}Andreas Knauf, \textquotedblleft Number theory, dynamical
systems and statistical mechanics,\textquotedblright\ \textit{\ Rev. Math.
Phys.} \textbf{11}(8), 1027--1060 (1999).

\bibitem {Igusa-old}J.-I. Igusa, \textit{Forms of higher degree} (Narosa
Publishing House, 1978).

\bibitem {Igusa}J.-I., Igusa, \textit{An introduction to the theory of local
zeta functions} (American Mathematical Society, International Press, 2000).

\bibitem {Kang et al}Li Kang Ming-Hsuan, Winnie Wen-Ching and Chian-Jen Wang,
\textquotedblleft The zeta functions of complexes from PGL(3): a
representation-theoretic approach,\textquotedblright\ \textit{Israel J. Math.}
\textbf{177}, 335--348 (2010).

\bibitem {Khrennikov et al}A. Yu. Khrennikov, F. M. Mukhamedov and J. F. F.
Mendes, \textquotedblleft On $p$-adic Gibbs measures of the countable state
Potts model on the Cayley tree,\textquotedblright\ \textit{Nonlinearity}
\textbf{20}(12), 2923--2937 (2007).

\bibitem {Khrenikov1}A. Khrennikov, \textit{Information Dynamics in Cognitive,
Psychological, Social and Anomalous Phenomena} (Springer, 2004).

\bibitem {Khrennikov2}A. Yu. Khrennikov, \textit{Non-Archimedean Analysis:
Quantum Paradoxes, Dynamical Systems and Biological Models} (Kluwer Academic
Publishers, 1997).

\bibitem {Khrennikov-Kozyrev}A. Yu. Khrennikov, S.V. Kozyrev,
\textquotedblleft Replica symmetry breaking related to a general ultrametric
space I: Replica matrices and functionals,\textquotedblright\ \textit{Physica
A: Statistical Mechanics and its Applications}, \textbf{359}, 222-240 (2006).

\bibitem {KKZuniga}A. Khrennikov, S. Kozyrev and W. A. Z\'{u}\~{n}iga-Galindo,
\textit{Ultrametric Equations and its Applications} (Cambridge University
Press, 2018).

\bibitem {Lerner-1}\`{E}. Yu Lerner, \textquotedblleft Feynman integrals of a
$p$-adic argument in a momentum space. I. Convergence,\textquotedblright%
\ \textit{Theoret. and Math. Phys.} \textbf{102}(3), 267--274 (1995).

\bibitem {Lerner-Missarov}\`{E}. Yu. Lerner and M. D. Missarov,
\textquotedblleft$p$-adic Feynman and string amplitudes,\textquotedblright%
\ \textit{Comm. Math. Phys.} \textbf{121}(1), 35--48 (1989).

\bibitem {Lerner-Missarov-2}\`{E}. Yu. Lerner and M. D. \ Missarov,
\textquotedblleft Scalar models of $p$-adic quantum field theory, and a
hierarchical model,\textquotedblright\ \textit{Theoret. and Math. Phys.}
\textbf{78}(2), 177--184 (1989).

\bibitem {Loeser}F. Loeser, \textquotedblleft Fonctions z\^{e}ta locales
d'Igusa \`{a} plusieurs variables, int\'{e}gration dans les fibres, et
discriminants,\textquotedblright\ \textit{Ann. Sci. \'{E}cole Norm. Sup.}
\textbf{22}(3), 435--471 (1989).

\bibitem {Mis}M. D. Missarov, \textquotedblleft The continuum limit in the
fermionic hierarchical model,\textquotedblright\ \textit{Theoret. and Math.
Phys.} \textbf{118}(1), 32--40 (1999).

\bibitem {Mukhamedov-1}F. Mukhamedov and H. Ak\i n, \textquotedblleft Phase
transitions for $p$-adic Potts model on the Cayley tree of order
three,\textquotedblright\ \textit{J. Stat. Mech. Theory Exp.} \textbf{7},
P07014, (2013).

\bibitem {Mukhamedov-2}F. Mukhamedov, \textquotedblleft On the strong phase
transition for the one-dimensional countable state $p$-adic Potts
model,\textquotedblright\ \textit{J. Stat. Mech. Theory Exp.} \textbf{1},
P01007, (2014).

\bibitem {Mukhamedov-3}F. Mukhamedov, M. Saburov and O. Khakimov,
\textquotedblleft On $p$-adic Ising-Vannimenus model on an arbitrary order
Cayley tree,\textquotedblright\ \textit{J. Stat. Mech. Theory Exp.}
\textbf{5}, P05032, (2015).

\bibitem {M-P-V}M. M\'{e}zard, G. Parisi, M. A. Virasoro, \textit{Spin glass
theory and beyond} (World Scientific, 1987).

\bibitem {Parisi}G. Parisi, \textquotedblleft On p-adic functional
integrals,\textquotedblright\ \textit{Modern Physics Letters A},
\textbf{3}(6), 639-643 (1988).

\bibitem {Parisi-Sourlas}G. Parisi and N. Sourlas, \textquotedblleft$p$-Adic
numbers and replica symmetry breaking,\textquotedblright\ \textit{Eur. Phys.
J. B} \textbf{14}, 535--542 (2000). https://doi.org/10.1007/s100510051063.

\bibitem {R-T-V}R. Rammal, G. Toulouse and M. A. Virasoro, \textquotedblleft
Ultrametricity for physicists,\textquotedblright\ \textit{Rev. Modern Phys.}
\textbf{58}(3), 765--788 (1986).

\bibitem {Rodriguez-Zuniga-2010}J. J. Rodr\'{\i}guez-Vega and W. A.
Z\'{u}\~{n}iga-Galindo, \textquotedblleft Elliptic pseudodifferential
equations and Sobolev spaces over p-adic fields,\textquotedblright%
\ \textit{Pacific J. Math. }\textbf{246}(2), 407--420 (2010).

\bibitem {Simon}Simon Barry, \textit{The statistical mechanics of lattice
gases. Vol. I.} (Princeton University Press, 1993)

\bibitem {Sinai}Ya. G. Sina\u{\i} , \textit{Theory of phase transitions:
rigorous results} (Pergamon Press 1982).

\bibitem {Sinclair}Ch. D. Sinclair and J. D. Vaaler, \textquotedblleft The
distribution of non-archimedean absolute Vandermonde
determinants,\textquotedblright\ \textit{Preprint, (}2019).

\bibitem {Taibleson}M. H. Taibleson, \textit{Fourier analysis on local fields}
(Princeton University Press, 1975).{}

\bibitem {Terras}Audrey Terras, \textit{Zeta functions of graphs}. \textit{A
stroll through the garden }(Cambridge University Press, 2011).

\bibitem {Veys-Zuniga}Willem Veys and W. A. Z\'{u}\~{n}iga-Galindo,
\textquotedblleft Zeta functions and oscillatory integrals for meromorphic
functions,\textquotedblright\ \textit{Adv. Math.} \textbf{31}, 295--337 (2017).

\bibitem {V-V-Z}V. S. Vladimirov, I. V. Volovich and E. I. Zelenov,
$p$\textit{-adic analysis and mathematical physics} (World Scientific, 1994).

\bibitem {Zuniga-JpA}W. A. Z\'{u}\~{n}iga-Galindo, \textquotedblleft
Non-archimedean replicator dynamics and Eigen's paradox,\textquotedblright%
\ \textit{J. Phys. A} \textbf{51}(50), 505601 (2018).

\bibitem {Zuniga-Nonlinearity}W. A. Z\'{u}\~{n}iga-Galindo, \textquotedblleft
Non-Archimedean reaction-ultradiffusion equations and complex hierarchic
systems,\textquotedblright\ \textit{Nonlinearity} \textbf{31}(6), 2590--2616 (2018).

\bibitem {Zuniga-LNM-2016}W. A. Z\'{u}\~{n}iga-Galindo,
\textit{Pseudodifferential equations over non-Archimedean spaces} (Springer
Nature, 2016 ).

\bibitem {Zuniga-Torba}W. A. Z\'{u}\~{n}iga-Galindo and Sergii M. Torba,
\textquotedblleft Non-Archimedean Coulomb gases,\textquotedblright\textit{ J.
Math. Phys.} \textbf{6}(1), 013504 (2020).
\end{thebibliography}
\end{document}